\documentclass[reqno,11pt]{amsart}
\usepackage{graphicx}
\usepackage{amssymb}
\usepackage{slashed}
\usepackage{enumitem}
\usepackage[usenames,dvipsnames]{pstricks}
\usepackage[mathscr]{eucal}
\numberwithin{equation}{section}
\allowdisplaybreaks[1]

\newcommand{\SetFigFont}[3]{}

\title[A Hamiltonian Formulation of Causal Variational Principles]{A Hamiltonian Formulation of \\ Causal Variational Principles}
\author[F.\ Finster]{Felix Finster}
\author[J.\ Kleiner]{Johannes Kleiner \\ \\ December 2016}
\address{Fakult\"at f\"ur Mathematik \\ Universit\"at Regensburg \\ D-93040 Regensburg \\ Germany}
\email{finster@ur.de, johannes.kleiner@ur.de}
\newtheorem{Def}{Definition}[section]
\newtheorem{Thm}[Def]{Theorem}
\newtheorem{Prp}[Def]{Proposition}
\newtheorem{Lemma}[Def]{Lemma}
\newtheorem{Remark}[Def]{Remark}
\newtheorem{Corollary}[Def]{Corollary}
\newtheorem{Example}[Def]{Example}

\newcommand{\Thanks}{\vspace*{.5em} \noindent \thanks}
\newcommand{\beq}{\begin{equation}}
\newcommand{\eeq}{\end{equation}}
\newcommand{\Proof}{\begin{proof}}
\newcommand{\QED}{\end{proof} \noindent}
\newcommand{\QEDrem}{\ \hfill $\Diamond$}

\newcommand{\la}{\langle}
\newcommand{\ra}{\rangle}

\newcommand{\C}{\mathbb{C}}
\newcommand{\R}{\mathbb{R}}

\newcommand{\Z}{\mathbb{Z}}
\newcommand{\N}{\mathbb{N}}

\newcommand{\F}{\mathscr{F}} % the free gauge group

\newcommand{\itemD}{\item[{\raisebox{0.125em}{\tiny $\blacktriangleright$}}]}

\newcommand{\D}{{\mathscr{D}}}

\DeclareFontFamily{OT1}{rsfso}{}
\DeclareFontShape{OT1}{rsfso}{m}{n}{ <-7> rsfso5 <7-10> rsfso7 <10-> rsfso10}{}
\DeclareMathAlphabet{\mycal}{OT1}{rsfso}{m}{n}

\usepackage{xcolor}

\DeclareMathOperator{\tr}{tr}

\DeclareMathOperator{\supp}{supp}
\renewcommand{\O}{{\mathscr{O}}}
\renewcommand{\L}{{\mathcal{L}}}
\newcommand{\Sact}{{\mathcal{S}}}
\newcommand{\T}{{\mathcal{T}}}

\newcommand\calB{{\mathcal{B}}}

\newcommand{\J}{\mathfrak{J}}

\newcommand{\Jdiff}{\mathfrak{J}^\text{\rm{\tiny{diff}}}}
\newcommand{\Jtest}{\mathfrak{J}^\text{\rm{\tiny{test}}}}
\newcommand{\Jlin}{\mathfrak{J}^\text{\rm{\tiny{lin}}}}
\newcommand{\Gdiff}{\Gamma^\text{\rm{\tiny{diff}}}}
\newcommand{\Gtest}{\Gamma^\text{\rm{\tiny{test}}}}
\newcommand{\Ctest}{C^\text{\rm{\tiny{test}}}}
\renewcommand{\u}{\mathfrak{u}}
\renewcommand{\v}{\mathfrak{v}}

\renewcommand{\H}{\mathscr{H}}
\newcommand{\Lin}{\text{\rm{L}}}
\definecolor{darkblue}{RGB}{0,91,163} %Farbe, um Änderungen zu markieren.
\begin{document}
\maketitle

\begin{abstract} 
Causal variational principles, which are the analytic core of the physical theory of causal fermion systems, 
are found to have an underlying Hamiltonian structure, giving a formulation
of the dynamics in terms of physical fields in space-time.
After generalizing causal variational principles to a class of lower semi-continuous Lagrangians on a smooth,
possibly non-compact manifold, the corresponding Euler-Lagrange equations are derived.
In the first part, it is shown under additional smoothness assumptions that the space of solutions of the
Euler-Lagrange equations has the structure of a symplectic Fr\'echet manifold.
The symplectic form is constructed as a surface layer integral which is shown to be invariant under
the time evolution. In the second part, the results and methods are extended to the non-smooth setting.
The physical fields correspond to variations of the universal measure described infinitesimally
by one-jets. Evaluating the Euler-Lagrange equations weakly, we derive linearized field equations for
these jets. In the final part, our constructions and results are illustrated in a detailed example
on $\mathbb{R}^{1,1} \times S^1$ where a local minimizer is given by a measure supported on a two-dimensional lattice.
\end{abstract}

\tableofcontents

\section{Introduction}\label{Intr}
The theory of causal fermion systems is an approach to describe fundamental physics. Giving
quantum mechanics, general relativity and quantum field theory as limiting cases, it is a candidate for a unified
physical theory (see~\cite{cfs} or the survey article~\cite{dice2014}).
In the present paper, we introduce a formalism to describe the dynamics of causal fermion systems
in terms of a Hamiltonian time evolution with a conserved symplectic form. This formulation has the major
advantage that it is closer to the 
conventional formulation of physics, making it possible to extend methods
and concepts from classical physics and symplectic geometry to the setting of causal fermion systems.
Our formalism is a suitable starting point for getting the connection
to the canonical formulation of quantum field theory in Fock spaces~\cite{qftlimit}
and for working out physical applications,
with the ultimate goal of making experimental predictions in the form of corrections to
measurable physical quantities. Furthermore, it sets the stage
for getting a connection to continuous spontaneous collapse models~\cite{jet3}.

In this introduction, we outline the main ideas and results of this paper in a non-technical way. 
The basic object of the theory of causal fermion systems is a measure~$\rho$ on a set of linear
operators on a Hilbert space (see~\cite{dice2014} or~\cite[Section~1.1]{cfs}).
Here we consider the more general and at the same time
easier accessible setting that~$\rho$ is a measure on a smooth, possibly non-compact manifold~$\F$
(for the detailed connection to causal fermion systems see Section~\ref{seccfs} below).
The {\em{causal variational principle}} is to minimize the {\em{causal action}}~$\Sact$ given by
\[ \Sact(\rho) = \int_\F d\rho(x) \int_\F d\rho(y)\: \L(x,y) \]
under variations of the measure~$\rho$, keeping the total volume fixed
(for details see Section~\ref{secnoncompact} below).
Here~$\L: \F \times \F \rightarrow \R_0^+$ is the Lagrangian.
While in the theory of causal fermion systems this Lagrangian is a specific function
(see~\cite[\S1.1.1]{cfs} or~\eqref{Ldef} below),
our results hold for general~$\L$ satisfying suitable regularity assumptions: smoothness in Section~\ref{SecSmooth} and lower semi-continuity in Sections~\ref{seccvpcfs} and~\ref{lowerSemiCont}.

Let~$\rho$ be a minimizer of the causal action principle (for mathematical details
see again Section~\ref{secnoncompact} below).
The key step towards describing the causal variational principle in terms of a Hamiltonian time evolution
is to consider variations of $\rho$ described by
a diffeomorphism $F: \F \rightarrow \F$ and a weight function $f: \F \rightarrow \R_0^+$.
More precisely, we consider families $(F_\tau)_{\tau \in \R}$ and $(f_\tau)_{\tau \in \R}$
of such diffeomorphisms and weight functions and
form a corresponding family~$(\rho_\tau)_{\tau \in \R}$ of measures by
\begin{align}\label{IntrRhoTau}
\rho_\tau = (F_\tau)_* \big( f_\tau \, \rho \big) \, .
\end{align}
Here $(F_\tau)_*\mu$ denotes the push-forward of the measure $\mu$ (defined 
for a subset~$\Omega \subset \F$ by~$((F_\tau)_*\mu)(\Omega)
= \mu ( F_\tau^{-1} (\Omega))$; see for example~\cite[Section~3.6]{bogachev}).
Thus the measure $\rho_\tau$ is obtained from $\rho$ by first multiplying with the weight function $f_\tau$ and then ``transporting'' the resulting measure with the diffeomorphism $F_\tau$ on $\F$.
Infinitesimal versions of the variations~\eqref{IntrRhoTau} consist of a scalar part corresponding to the
$\tau$-derivative of $f_\tau$ and a vectorial part corresponding to the $\tau$-derivative of $F_\tau$.
Thus variations of the form~\eqref{IntrRhoTau} can be described infinitesimally by 
a pair~$(a,v)$ of a real-valued function and a vector field.
In our formulation, physical fields in space-time are described in terms of such pairs~$(a,v)$.
Thus pairs~$(a,v)$ can be viewed as generalized physical fields.
In order to have a short name which cannot be confused with common notions in physics,
we refer to the pairs~$(a,v)$ as {\em{jets}}, being elements of the
corresponding {\em{jet space}}\footnote{The connection to jets in differential geometry
(see for example~\cite{saunders}) is obtained by considering real-valued functions on~$\F$.
Then their one-jets are elements in~$C^\infty(\F) \oplus \Gamma(\F, T^*\F)$. Identifying the
cotangent space with the tangent space gives our jet space~$\J$.}
\[ \J := \big\{ \v = (b,v) \text{ with } b : \F \rightarrow \R \text{ and } v \in \Gamma(\F) \} \]
(see~\eqref{Jdef} and~\eqref{JDiffLip}).

We are interested in variations of the form~\eqref{IntrRhoTau} which are minimizers of the causal action also
for~$\tau \neq 0$. Such ``families of minimizers'' are of interest because 
in the theory of causal fermion systems, they
correspond to variations which satisfy the physical equations.
The requirement of~$\rho_\tau$ being a minimizer for all~$\tau$ gives rise to
conditions for the jet~$\v = \partial_\tau \rho_\tau|_{\tau=0} \in \J$ describing the infinitesimal
variation. In view of the similarities and the correspondence to classical field theory
(as worked out in~\cite[\S1.4.1 and Chapters~3-5]{cfs} and~\cite{perturb}),
we refer to these conditions as the \textit{linearized field equations}.
The linearized field equations can be written as (for details see Lemma~\ref{lemmalin} and~\eqref{eqlinlip})
\begin{align}\label{IntreqLinFieldEq}
\nabla_{\u} \bigg( \int_M \big( \nabla_{1, \v} + \nabla_{2, \v} \big) \L(x,y)\: d\rho(y) - \nabla_\v \:\frac{\nu}{2} \bigg)
= 0
\end{align}
for all~$\u \in \Jtest$ and~$x\in M$.
Here we used the following notions and definitions:
\begin{itemize}[leftmargin=1.5em]
\itemD {\em{Space-time}}~$M$ is defined as the support of the universal measure,
\[ M := \supp \rho \subset \F \:. \]
In the setting of causal fermion systems, this definition indeed
generalizes the usual notion of space-time (being Minkowski space or a Lorentzian manifold;
for details see~\cite[Section~1.2]{cfs} or~\cite[Sections~4 and~5]{lqg}).
\itemD The {\em{test jets}}~$\Jtest \subset \J$ are defined as a subspace of one-jets used to test the 
requirement of minimality in a weak sense. This so-called {\em{weak evaluation of the Euler-Lagrange (EL)
equations}} is an important mathematical and physical concept because
by choosing~$\Jtest$ appropriately, one can restrict attention to the part of the information contained
in the EL equations which is relevant for the application in mind.
In order to illustrate how this works, we give a typical example:
For the description of macroscopic physics, one would like to
disregard effects which come into play only on the Planck scale. To this end,
one chooses~$\Jtest$ as a space of jets which vary only on the macroscopic scale,
so that ``fluctuations on the Planck scale are filtered out''
(for details see~\cite{perturb}).
\itemD The {\em{derivative}}~$\nabla_{\v}$ in the direction of a one-jet~$\v = (b, v)$ is defined
as a combination of multiplication and differentiation,
\[ \nabla_{\v} \eta(x)= b(x) \,\eta (x) + D_v \eta(x) \]
(where~$D_v$ is the usual directional derivative of functions on~$\F)$.
Likewise, the derivatives~$\nabla_{1, \v}$ and~$\nabla_{2, \v}$ denote partial derivatives acting on the
first and second argument of~$\L(x,y)$, respectively.
\itemD The parameter~$\nu \geq 0$ is the Lagrange multiplier corresponding to the volume constraint.
\end{itemize}

In the smooth setting of Section~\ref{SecSmooth}, we consider the set $\calB$ of all measures of the form~\eqref{IntrRhoTau} which satisfy the weak EL equations (for details see~\eqref{ELweak2})
\[ \nabla_{\u} \bigg( \int_\F \L(x,y)\: d\rho(y) - \frac{\nu}{2} \bigg) = 0 \qquad \text{for all~$\u \in \J$
and~$x \in M$}\:. \]
We assume that this set is a smooth Fr{\'e}chet manifold. 
Then infinitesimal variations of~\eqref{IntrRhoTau} which are solutions of~\eqref{IntreqLinFieldEq} are 
vectors of the tangent space~$T_\rho \calB$.
In this setting, we show that the structure of the causal variational principle
gives rise to a {\em{symplectic form}} on $\calB$ (Section~\ref{SecSympForm}).
Namely, for any~$\u, \v \in T_\rho \calB$ and~$x,y \in M$, let
\[ \sigma_{\u, \v}(x,y) := \nabla_{1,\u} \nabla_{2,\v} \L(x,y) - \nabla_{1,\v} \nabla_{2,\u} \L(x,y) \:. \]
Then, given a compact subset~$\Omega \subset \F$, we define the bilinear form
\beq \label{IntrOSI}
\sigma_\Omega \::\: T_\rho \calB \times T_\rho \calB \rightarrow \R\:,
\qquad \sigma_\Omega(\u, \v) = \int_\Omega d\rho(x) \int_{M \setminus \Omega} d\rho(y)\:
\sigma_{\u, \v}(x,y) \:.
\eeq
This is an example of a {\em{surface layer integral}} as first introduced in~\cite{noether}.
The structure of such surface layer integrals can be understood most easily 
in the special situation that the Lagrangian is of short range
in the sense that~$\L(x,y)$ vanishes unless~$x$ and~$y$ are close together.
In this situation, we only get a contribution to the double integral~\eqref{IntrOSI}
if both~$x$ and~$y$ are close to the boundary~$\partial \Omega$.
With this in mind, surface layer integrals can be understood as an adaptation
of surface integrals to the setting of causal variational principles
(for a more detailed explanation see~\cite[Section~2.3]{noether}).
In~\cite{noether}, it is shown that there are conservation laws expressed in terms of
surface layer integrals which in the continuum limit reduce to the well-known charge and current conservation
laws expressed in terms of surface integrals. Here, we prove a different conservation law which
makes it possible to introduce a symplectic form and a Hamiltonian time evolution (for details see Theorem~\ref{thmOSI}).

\begin{itemize}[leftmargin=5.3em]
\item[\bf{Theorem.}]
For any compact subset~$\Omega \subset \F$, the surface layer integral~\eqref{IntrOSI}
vanishes for all~$\u, \v \in T_\rho \calB$.
\end{itemize}

This theorem has the following connection to conservation laws.
Let us assume that~$M$ admits a sensible notion of ``spatial infinity'' and
that the jets~$\u, \v \in T_\rho \calB$ have suitable decay properties at spatial infinity.
Then one can chose a sequence~$\Omega_n \subset M$ of compact sets
which form an exhaustion of a set~$\Omega$ which extends up to spatial infinity
(see Figure~\ref{figjet1}~(a) and~(b)).
\begin{figure}
% \usepackage[usenames,dvipsnames]{pstricks}
% \usepackage{epsfig}
% \usepackage{pst-grad} % For gradients
% \usepackage{pst-plot} % For axes
% \usepackage[space]{grffile} % For spaces in paths
% \usepackage{etoolbox} % For spaces in paths
% \makeatletter % For spaces in paths
% \patchcmd\Gread@eps{\@inputcheck#1 }{\@inputcheck"#1"\relax}{}{}
% \makeatother
% \psscalebox{1.0 1.0} % Change this value to rescale the drawing.
{
\begin{pspicture}(0,-0.5682992)(16.27295,1.0682992)
\definecolor{colour0}{rgb}{0.8,0.8,0.8}
\pspolygon[linecolor=colour0, linewidth=0.02, fillstyle=solid,fillcolor=colour0](9.348506,-0.5327436)(9.348506,1.0361453)(9.63295,0.92947865)(9.926284,0.8805898)(10.339617,0.84058976)(10.837395,0.85836756)(11.326283,0.9072564)(11.668506,0.93836755)(11.988505,0.96058977)(12.184061,0.94281197)(12.801839,0.9961453)(13.241839,1.0050342)(13.259617,-0.537188)(12.948505,-0.50163245)(12.539617,-0.50163245)(12.11295,-0.48385468)(11.748506,-0.497188)(11.459617,-0.47496578)(11.246284,-0.457188)(10.9529505,-0.46163246)(10.717395,-0.4349658)(10.406283,-0.42163244)(9.997395,-0.39941022)(13.259617,-0.5327436)
\pspolygon[linecolor=colour0, linewidth=0.02, fillstyle=solid,fillcolor=colour0](4.5396166,-0.27941024)(4.530728,0.9872564)(4.815172,0.8805898)(5.1085057,0.83170086)(5.530728,0.7961453)(5.9573946,0.7872564)(6.4596167,0.8139231)(6.819617,0.84947866)(7.0773945,0.8761453)(7.3662834,0.8939231)(7.9840612,0.94725645)(8.424061,0.95614535)(8.424061,-0.31052133)(8.03295,-0.33274356)(7.624061,-0.3460769)(7.3040614,-0.38163245)(7.0373945,-0.40385467)(6.779617,-0.4438547)(6.468506,-0.47496578)(6.1885056,-0.50163245)(5.8773947,-0.51052135)(5.55295,-0.497188)(5.179617,-0.4482991)(4.8462834,-0.377188)
\pspolygon[linecolor=colour0, linewidth=0.02, fillstyle=solid,fillcolor=colour0](0.03295012,0.12947866)(0.19739456,0.24503422)(0.35739458,0.3250342)(0.57517236,0.43614534)(0.9485057,0.57836753)(1.3885057,0.6939231)(1.8151723,0.74725646)(2.121839,0.7605898)(2.43295,0.742812)(2.7973945,0.6539231)(3.121839,0.542812)(3.441839,0.3917009)(3.681839,0.25836754)(3.5973945,0.13836755)(3.39295,0.013923102)(3.130728,-0.123854674)(2.8551724,-0.24385467)(2.4773946,-0.31941023)(2.161839,-0.36829913)(1.8151723,-0.38163245)(1.5173945,-0.37274358)(1.081839,-0.30607688)(0.62850565,-0.16385467)(0.281839,-0.0038546752)
\rput[bl](1.8151723,0.025034213){$\Omega_n$}
\rput[bl](6.0973945,0.12281199){\normalsize{$\Omega$}}
\psbezier[linecolor=black, linewidth=0.04](0.02406123,0.12281199)(0.7173699,0.6167855)(1.5001423,0.76261884)(2.0707278,0.7672564358181418)(2.6413136,0.77189404)(3.0712237,0.6218133)(3.7240613,0.242812)
\psbezier[linecolor=black, linewidth=0.04](0.009616784,0.12947866)(0.52181435,-0.15099226)(0.95347565,-0.34627005)(1.5962834,-0.38941023084852644)(2.2390912,-0.4325504)(3.124557,-0.3104089)(3.709617,0.24947865)
\psbezier[linecolor=black, linewidth=0.04](4.525172,-0.297188)(5.046259,-0.4843256)(5.7001424,-0.5373812)(6.1696167,-0.518299119737415)(6.6390915,-0.49921706)(7.315668,-0.33707556)(8.438506,-0.3282991)
\psbezier[linecolor=black, linewidth=0.04](4.520728,0.9917009)(5.0418143,0.8045633)(5.5401425,0.7915077)(5.991839,0.783923102484807)(6.443536,0.7763385)(7.3112235,0.95181334)(8.434061,0.96058977)
\psbezier[linecolor=black, linewidth=0.04](9.338506,1.0494787)(9.859592,0.8623411)(10.375698,0.85372996)(10.751839,0.8683675469292492)(11.12798,0.88300514)(12.129002,1.0095911)(13.251839,1.0183675)
\rput[bl](10.97295,0.10947866){\normalsize{$\Omega_N$}}
\rput[bl](1.7040613,-1.0682992){(a)}
\rput[bl](6.1751723,-1.0638547){(b)}
\rput[bl](11.059617,-1.0638547){(c)}
\rput[bl](7.2085056,-0.23718801){\normalsize{$N_1$}}
\rput[bl](7.2351723,0.482812){\normalsize{$N_2$}}
\rput[bl](12.319616,0.5850342){\normalsize{$N$}}
\end{pspicture}
}
\caption{Choices of space-time regions.}
\label{figjet1}
\end{figure}
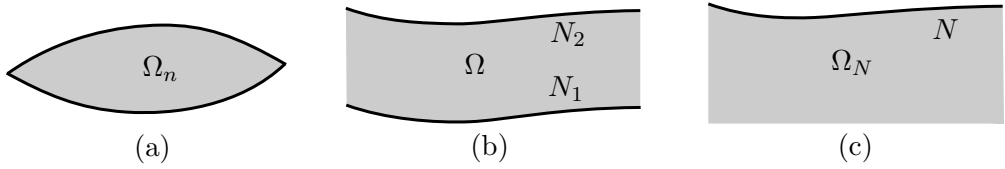%
Considering the surface layer integrals~\eqref{IntrOSI} for~$\Omega_n$ and passing to limit,
one concludes that also the surface layer integral corresponding to~$\Omega$ vanishes.
Let us assume that the boundary~$\partial \Omega$ has two components~$N_1$
and~$N_2$ (as in Figure~\ref{figjet1}~(b)).
Then the above theorem implies that the surface layer integrals over~$N_1$ and~$N_2$ coincide
(where the surface layer integral over~$N$ is defined as the surface layer integral
corresponding to a set~$\Omega_N$ with~$\partial \Omega_N = N$ as shown in Figure~\ref{figjet1}~(c)).
In other words, the quantity
\begin{align}\label{IntrOSIN}
 \sigma_{\Omega_N}(\u, \v) = \int_{\Omega_N} d\rho(x) \int_{M \setminus {\Omega_N}} d\rho(y) \:\sigma_{\u, \v}(x,y)
\end{align}
is well-defined and independent of the choice of~$N$.
In this setting, the surfaces~$N$ can be interpreted as Cauchy surfaces,
and the conservation law~\eqref{IntrOSIN} means that the bilinear form~$\sigma_{\Omega_N}$
is preserved under the time evolution. This is what we mean by
\textit{Hamiltonian time evolution}.

To avoid misunderstandings, we point out that, in contrast to classical field theory, in our setting
the time evolution is {\em{not defined infinitesimally}} by a Hamiltonian or a Hamiltonian vector field
(this is obvious from the fact that causal fermion systems allow for the description of discrete space-times,
where a continuous time evolution makes no sense).
Instead, the time evolution should be thought of as a mapping from the jets
in a surface layer around~$N_1$ to the jets in a surface layer around~$N_2$.
This mapping is a {\em{symplectomorphism}} with regard
to~$\sigma_{\Omega_{N_1}}$ and~$\sigma_{\Omega_{N_2}}$, respectively.
For clarity, we also note that it is essential that~$\F$ and~$M$ are non-compact because otherwise, Theorem~\ref{thmOSI} would immediately imply
that~$\sigma_{\Omega_N}(\u, \v) \equiv 0$. For non-compact $M$, however,
Theorem~\ref{thmOSI} only applies to the
difference~$\sigma_{\Omega_{N_1}}\! - \sigma_{\Omega_{N_2}}$, thereby giving a conservation law
for a non-trivial surface layer integral.

The independence of~\eqref{IntrOSIN} from $N$ allows us to define a bilinear form $\sigma$ on $\calB$ by
\[ \sigma \::\: T_\rho \calB \times T_\rho \calB \rightarrow \R\:,
\qquad \sigma(\u, \v) := \sigma_{\Omega_N}(\u, \v) \]
for an arbitrary choice of~$N$. This bilinear form turns out to be closed
(see Lemma~\ref{lemmaclosed}), thus defining a {\em{presymplectic form}} on $\calB$. 
Finally, by restricting~$\sigma$ to a suitable subspace of~$T_\rho \calB$
we arrange that~$\sigma$ is non-degenerate, giving a {\em{symplectic form}} (see Section~\ref{SecSympForm}).

As mentioned above, in the theory of causal fermion systems the Lagrangian is not smooth, but
merely Lipschitz-continuous (see~\cite[Section~1.1]{cfs} and~\cite{support}).
In order to cover this situation, in Section~\ref{lowerSemiCont} we treat the more general case of a
{\em{lower semi-continuous}} Lagrangian. Our main motivation for this generalization is
that many simple examples are easier to state if we allow for discontinuities of the Lagrangian.
After defining jet spaces as infinitesimal versions of families of solutions similar as described above
(see Sections~\ref{WeakElLSC} and~\ref{SmoothVar}) and establishing the necessary conditions for the linearized field equations to be well-defined (see Definition~\ref{deflin}),
we again establish a conservation law for the bilinear
form $\sigma_\Omega(\u,\v)$ (see Theorem~\ref{thmOSIlip}).

In Section~\ref{seclattice}, we illustrate our constructions in an example
which is simple enough for an explicit analysis but nevertheless
captures some features of a physical field theory.
We choose $\L(x,y)$  in such a way that the minimizing measure is
supported on a two-dimensional lattice. This reflects a general feature
of the theory of causal fermion systems that space-time ``discretizes itself''
on the Planck scale, thus avoiding the ultraviolet divergences of quantum field theory
(see~\cite[Section~4]{rrev}).
The structure of the minimizers of this model is reminiscent
of a nonlinear sigma model on a lattice in $\R^{1,1}$ with values in~$S^1$.
We solve the linearized field equations and construct the symplectic form.
This example also serves as the starting point for a numerical exploration of the connection to dynamical
collapse theories in~\cite{jet3}.

\section{Causal Variational Principles and Causal Fermion Systems} \label{seccvpcfs}
In this section we review and generalize the setting of causal variational principles
and recall a few definitions and basic results which will be needed later on.
Thus this section provides the preliminaries needed for the construction of the Hamiltonian time evolution.
After introducing causal variational principles in the non-compact setting (Section~\ref{secnoncompact}),
we derive the corresponding Euler-Lagrange equations (Section~\ref{secEL}).
The connection to the theory of causal fermion systems is established in Section~\ref{seccfs}.

\subsection{Causal Variational Principles in the Non-Compact Setting} \label{secnoncompact}
We now introduce causal variational principles in the non-compact setting
(for the simpler compact setting see~\cite{continuum, support, noether}).
Let~$\F$ be a (possibly non-compact) smooth manifold of dimension~$m \geq 1$
and~$\rho$ a (positive) Borel measure on~$\F$ (the {\em{universal measure}}).
Moreover, we are given a non-negative function~$\L : \F \times \F \rightarrow \R^+_0$
(the {\em{Lagrangian}}) with the following properties:
\begin{itemize}[leftmargin=2em]
\item[(i)] $\L$ is symmetric: $\L(x,y) = \L(y,x)$ for all~$x,y \in \F$.\label{Cond1}
\item[(ii)] $\L$ is lower semi-continuous, i.e.\ for all sequences~$x_n \rightarrow x$ and~$y_{n'} \rightarrow y$,
\[ \L(x,y) \leq \liminf_{n,n' \rightarrow \infty} \L(x_n, y_{n'})\:. \]\label{Cond2}
\end{itemize}
If the total volume~$\rho(\F)$ is finite, the {\em{causal variational principle}} is to minimize the action
\beq \label{Sact} 
\Sact (\rho) = \int_\F d\rho(x) \int_\F d\rho(y)\: \L(x,y) 
\eeq
under variations of the measure~$\rho$, keeping the total volume~$\rho(\F)$ fixed
({\em{volume constraint}}).
If~$\rho(\F)$ is infinite, it is not obvious how to implement the volume constraint,
making it necessary to proceed as follows:
First, we make the following additional assumptions:
\begin{itemize}[leftmargin=2em]
\item[(iii)] The measure~$\rho$ is {\em{locally finite}}
(meaning that any~$x \in \F$ has an open neighborhood~$U$ with~$\rho(U)< \infty$).\label{Cond3}
\item[(iv)] The function~$\L(x,.)$ is $\rho$-integrable for all~$x \in \F$, giving
a lower semi-continuous and bounded function on~$\F$. \label{Cond4}
\end{itemize}
We remark that, since a manifold is second countable, property~(iii) implies that~$\rho$
is $\sigma$-finite.
In view of the computations later in this paper,
it is most convenient to subtract a constant~$\nu/2$ from the integral over~$\L(x,.)$
by introducing the function
\beq \label{elldef}
\ell(x) = \int_\F \L(x,y)\: d\rho(y) - \frac{\nu}{2} \::\: \F \rightarrow \R \quad
\text{bounded and lower semi-continuous}\:,
\eeq
where the parameter~$\nu \in \R$ will be specified below.
We let~$\tilde{\rho}$ be another Borel measure on~$\F$
which satisfies the conditions
\beq \label{totvol}
\big| \tilde{\rho} - \rho \big|(\F) < \infty \qquad \text{and} \qquad
\big( \tilde{\rho} - \rho \big) (\F) = 0
\eeq
(where~$|.|$ denotes the total variation of a measure;
see~\cite[\S28]{halmosmt} or~\cite[Section~6.1]{rudin}).
Then the difference of the actions as given by
\beq \label{integrals}
\begin{split}
\big( &\Sact(\tilde{\rho}) - \Sact(\rho) \big) = \int_\F d(\tilde{\rho} - \rho)(x) \int_\F d\rho(y)\: \L(x,y) \\
&\quad + \int_\F d\rho(x) \int_\F d(\tilde{\rho} - \rho)(y)\: \L(x,y) 
+ \int_\F d(\tilde{\rho} - \rho)(x) \int_\F d(\tilde{\rho} - \rho)(y)\: \L(x,y)
\end{split}
\eeq
is well-defined in view of the following lemma.
\begin{Lemma} The integrals in~\eqref{integrals} are well-defined with values in~$\R \cup \{\infty\}$. Moreover,
\beq \label{integrals2}
\begin{split}
\big( \Sact(\tilde{\rho}) - \Sact(\rho) \big) &= 2 \int_\F \Big(\ell(x) + \frac{\nu}{2} \Big) \:d(\tilde{\rho} - \rho)(x) \\
&\quad + \int_\F d(\tilde{\rho} - \rho)(x) \int_\F d(\tilde{\rho} - \rho)(y)\: \L(x,y) \:.
\end{split}
\eeq
\end{Lemma}
\Proof Decomposing the signed measure~$\mu=\tilde{\rho}-\rho$ into its positive and negative parts,
$\mu = \mu^+-\mu^-$ (see the Jordan decomposition in~\cite[\S29]{halmosmt}), the measures~$\mu^\pm$
are both positive measures of finite total volume and~$\mu^- \leq \rho$.
In order to show that the integrals in~\eqref{integrals}
are well-defined, we need to prove that the negative contributions are finite, i.e.
\beq \label{finite}
\int_\F d\mu^-(x) \int_\F d\rho\: \L(x,y) < \infty \qquad \text{and} \qquad
\int_\F d\mu^+(x) \int_\F d\mu^-(y)\: \L(x,y) < \infty
\eeq
(here we apply Tonelli's theorem and make essential use of the fact that the Lagrangian is non-negative).
The bounds~\eqref{finite} follow immediately from the estimates
\begin{align*}
\int_\F & d\mu^-(x) \int_\F d\rho\: \L(x,y) = \int_\F d\mu^-(x) \:\Big(\ell(x) + \frac{\nu}{2} \Big)
\leq \Big( \sup_\F \ell + \frac{\nu}{2} \Big) \:\mu^-(\F) < \infty \\
\int_\F &d\mu^+(x) \int_\F d\mu^-(y)\: \L(x,y) \leq \int_\F d\mu^+(x) \int_\F d\rho\: \L(x,y) \\
&= \int_\F d\mu^+(x) \:\Big(\ell(x) + \frac{\nu}{2} \Big)\:
\leq \Big( \sup_\F \ell + \frac{\nu}{2} \Big) \:\mu^+(\F) < \infty \:,
\end{align*}
where we used the fact that~$\ell$ is assumed to be a bounded function on~$\F$.
\QED

\begin{Def} The measure~$\rho$ is said to be a {\bf{minimizer}} of the causal action
if the difference~\eqref{integrals2}
is non-negative for all~$\tilde{\rho}$ satisfying~\eqref{totvol},
\[ \big( \Sact(\tilde{\rho}) - \Sact(\rho) \big) \geq 0 \:. \]
\end{Def}

We close this section with a remark on the {\em{existence theory}}.
If~$\F$ is compact, the existence of minimizers can
be shown just as in~\cite[Section~1.2]{continuum} using the Banach-Alaoglu theorem
(the fact that~$\L$ is semi-continuous implies that the weak-$*$-limit of a minimizing sequence of
measures is indeed a minimizer).
In the non-compact setting, the existence theory has not yet been developed
(for more details on this point see~\cite[\S1.1.1]{cfs}). For the purpose of the present paper,
all we need is that the causal action principle admits {\em{local}} minimizers
which satisfy the corresponding Euler-Lagrange equations.
These concepts will be introduced in Sections~\ref{secEL} and~\ref{seclocmin} below.
Moreover, in Section~\ref{seclattice} we will analyze an example where local minimizers
exist although~$\F$ is non-compact.

\subsection{The Euler-Lagrange Equations} \label{secEL}
We now derive the Euler-Lagrange (EL) equations, following the method in the
compact setting~\cite[Lemma~3.4]{support}.
\begin{Lemma} (The Euler-Lagrange equations) \label{lemmaEL}
Let~$\rho$ be a minimizer of the causal action. Then
\beq \label{EL1}
\ell|_{\supp \rho} \equiv \inf_\F \ell \:.
\eeq
\end{Lemma}
\Proof Given~$x_0 \in \supp \rho$, we choose an open neighborhood~$U$ with~$0 < \rho(U)<\infty$.
For any~$y \in \F$ we consider the family of measures~$(\tilde{\rho}_\tau)_{\tau \in [0,1)}$ given by
\[ \tilde{\rho}_\tau = \chi_{M \setminus U} \,\rho + (1-\tau)\, \chi_U \, \rho + \tau\, \rho(U)\, \delta_{y} \]
(where~$\delta_y$ is the Dirac measure supported at~$y$). Then
\beq \label{tilderho}
\tilde{\rho}_\tau - \rho = -\tau\, \chi_U \, \rho + \tau\, \rho(U)\, \delta_{y} 
= \tau \big( \rho(U)\, \delta_{y} - \chi_U\, \rho \big) \:,
\eeq
implying that~$\tilde{\rho}_\tau$ satisfies~\eqref{totvol}. Hence
\begin{align*} 
0 &\leq \big(\Sact(\tilde{\rho}) - \Sact(\rho) \big) = 2 \tau 
\left( \rho(U)\, \Big( \ell(y) + \frac{\nu}{2} \Big)- \int_U \Big( \ell(x)+ \frac{\nu}{2} \Big)\, d\rho(x) \right) + \O \big(\tau^2 \big) \:.
\end{align*}
As a consequence, the linear term must be non-negative,
\begin{align} \label{ELInt}
\ell(y) \geq \frac{1}{\rho(U)} \int_U \ell(x)\, d\rho(x) \:.
\end{align}
Assume that~\eqref{EL1} is false. Then there is~$x_0 \in \supp \rho$ and~$y \in \F$ such that $\ell(x_0) > \ell(y)$. 
Lower semi-continuity of $\ell$ implies that there is an open
neighborhood $U$ of $x_0$ such that $\ell(x) > \ell(y)$ for all $x \in U$, in contradiction to~\eqref{ELInt}. This gives the result.
\QED
We always choose~$\nu$ such that~$\inf_\F \ell=0$. Then the EL equations~\eqref{EL1} simplify to
\begin{align}\label{ELstrong}
\ell|_{\supp \rho} \equiv \inf_\F \ell = 0 \: .
\end{align}
We remark that~$\nu$ can be understood as the Lagrange multiplier describing the volume constraint;
see~\cite[\S1.4.1]{cfs}.

\subsection{The Setting of Causal Fermion Systems} \label{seccfs}
We now explain how the causal action principle for causal fermion systems
(as introduced in~\cite[Section~1.1]{cfs}) can be described within the above setting
(for an introduction to causal fermion systems see~\cite[Chapter~1]{cfs}
or~\cite{dice2014}).
The main difference compared to the setting in Section~\ref{secnoncompact} is
that the causal action principle involves additional constraints, namely the
trace constraint and the boundedness constraint. We now explain how to
incorporate these constraints in a convenient way.
For a minimizer of the causal action, the local trace is constant on the support of the
universal measure (see~\cite[Proposition~1.4.1]{cfs}). With this in mind, we may
restrict attention to operators with fixed trace. When doing so, the trace constraint is
trivially satisfied. The boundedness constraint, on the other hand, can be incorporated by
Lagrange multiplier term. Finally, in the setting of causal fermion systems, the set~$\F$
is not necessarily a smooth manifold. In order to avoid this problem, we restrict attention
to minimizers for which all space-time points are regular (see~\cite[Definition~1.1.5]{cfs}).
Then we may restrict attention to operators which have exactly~$n$ positive and~$n$
negative eigenvalues. The resulting set of operators is a smooth manifold
(see the concept of a flag manifold in~\cite{helgason}). This leads us to the following setup:

Let~$(\H, \la .|. \ra_\H)$ be a finite-dimensional complex Hilbert space. Moreover, we are given
parameters~$n \in \N$ (the spin dimension), $c > 0$ (the constraint for the local trace)
and~$\kappa>0$ (the Lagrange multiplier of the boundedness constraint)\footnote{We
remark that the Lagrange multiplier~$\kappa$ is strictly positive because
otherwise there are no minimizers; see~\cite[Example~2.9]{continuum}
and~\cite[Exercise~1.4]{cfs}.}.
We let~$\F \subset \Lin(\H)$ be the set of all
self-adjoint operators~$F$ on~$\H$ with the following properties:
\begin{itemize}[leftmargin=2em]
\itemD $F$ has finite rank and (counting multiplicities) has~$n$ positive and~$n$ negative eigenvalues. \\[-0.8em]
\itemD The local trace is constant, i.e.
\[ \tr(F) = c\:. \]
\end{itemize}
On~$\F$ we consider the topology induced by the sup-norm on~$\Lin(\H)$.
For any~$x, y \in \F$, the product~$x y$ is an operator
of rank at most~$2n$. We denote its non-trivial eigenvalues counting algebraic multiplicities
by~$\lambda^{xy}_1, \ldots, \lambda^{xy}_{2n} \in \C$.
We introduce the Lagrangian by
\beq \label{Ldef}
\L(x,y) = \frac{1}{4n} \sum_{i,j=1}^{2n} \Big( \big|\lambda^{xy}_i \big| - \big|\lambda^{xy}_j \big| \Big)^2
+ \kappa\: \bigg( \sum_{i,j=1}^{2n} \big|\lambda^{xy}_i \big| \bigg)^2 \:.
\eeq
Clearly, this Lagrangian is continuous on~$\F \times \F$.
Moreover, since~$\kappa>0$, the Lagrangian is non-negative.

Therefore, we are back in the setting of
Section~\ref{secnoncompact}. The EL equations in Lemma~\ref{lemmaEL} agree with
the EL equations as derived for the causal action principle with constraints
in~\cite{lagrange} (see~\cite[Theorem~1.3]{lagrange}).

\section{The Symplectic Form in the Smooth Setting}\label{SecSmooth}
In order to introduce our concepts in the simplest possible setting,
in this section we assume that~$\F$ is a smooth manifold of dimension~$m \geq 1$
and that the Lagrangian~$\L \in C^\infty(\F \times \F, \R^+_0)$ is smooth.
Moreover, we let~$\rho$ be a regular Borel measure on~$\F$
which satisfies the EL equations \eqref{ELstrong} corresponding to the causal action in the sense that
the smooth function~$\ell$ defined by~\eqref{elldef},
\[ \ell(x) = \int_\F \L(x,y)\: d\rho(y) - \frac{\nu}{2} \in C^\infty(\F,\R_0^+) \:, \]
is minimal and vanishes on~$M$,
\beq \label{ELsmooth}
\ell|_{\supp \rho} \equiv \inf_\F \ell = 0 \:.
\eeq
The constructions in this section should be seen as a preparation
for the lower semi-continuous setting to be considered in Section~\ref{lowerSemiCont}.
Before going on, we remark that the value of the parameter~$\nu$ can be changed arbitrarily
by rescaling the measure according to
\[ \rho \rightarrow \lambda \rho \qquad \text{with} \qquad \lambda>0 \:. \]
Therefore, without loss of generality we can keep~$\nu$ fixed when varying or perturbing the measure.

\subsection{The Weak Euler-Lagrange Equations}
Clearly, the EL equations~\eqref{ELsmooth} imply the weaker equations
\beq \label{ELweak}
\ell|_M \equiv 0 \qquad \text{and} \qquad D \ell|_M \equiv 0
\eeq
(where~$D \ell(p) : T_p \F \rightarrow \R$ is the derivative).
In order to combine these two equations in a compact form, we introduce the
smooth one-jets
\beq \label{Jdef}
\J := \big\{ \u = (a,u) \text{ with } a \in C^\infty(\F, \R) \text{ and } u \in \Gamma(\F) \big\} \:,
\eeq
where~$\Gamma(\F)$ denotes the smooth vector fields on~$\F$.
Defining the derivative in direction of a one-jet by
\beq \label{Djet}
\nabla_{\u} \ell(x) := a(x)\, \ell(x) + \big(D_u \ell \big)(x) \:,
\eeq
we can write~\eqref{ELweak} as
\beq \label{ELweak2}
\nabla_{\u} \ell |_M \equiv 0 \qquad \text{for all~$\u \in \J$}\:.
\eeq
We refer to these equations as the {\em{weak EL equations}}.

\subsection{The Nonlinear Solution Space} \label{secfam}
Our next step is to analyze families of measures which satisfy the weak EL equations.
In order to obtain these families of solutions, we want to vary a given
measure~$\rho_0$ (not necessarily a minimizer) without changing its general structure.
To this end, we multiply~$\rho_0$ by a weight function
and apply a diffeomorphism, i.e.
\beq \label{rhoFf}
\rho = F_* \big( f \,\rho_0 \big) \:,
\eeq
where~$F : \F \rightarrow \F$ is a smooth diffeomorphism and~$f \in C^\infty(\F, \R^+)$.
We now consider a set of such measures which all satisfy the weak EL equations,
\beq
\calB \subset \left\{ \text{$\rho$ of the form~\eqref{rhoFf}} \:\big|\: \text{the weak EL equations~\eqref{ELweak2} are satisfied} \right\} \label{Bdef}
\eeq
(for fixed~$\rho_0$). 
We make further simplifying assumptions on $\calB$.
First, we assume that~$\calB$ is a smooth {\em{Fr{\'e}chet manifold}}
(endowed with the compact-open topology on~$C^\infty(\F, \R^+_0)$
and on the diffeomorphisms; for details see Appendix~\ref{appclosed}).
Then for~$\rho \in \calB$, a tangent vector~$\v \in T_\rho \calB$,
being an infinitesimal variation of the measures in~\eqref{rhoFf},
consists of a function~$b$ (describing the infinitesimal change of the weight) and
a vector field~$v$ (being the infinitesimal generator of the diffeomorphism).
Having chosen the Fr{\'e}chet topology such that~$b$ and~$v$ are smooth, we obtain a jet~$\v = (b,v) \in \J$.
Hence the tangent space can be identified with a subspace of the one-jets,
\[ \T_\rho \calB \subset \J \:. \]
A second assumption is needed in order to ensure that we can exchange integration with differentiation in the proof of Lemma~\ref{lemmalin} below.
To this end, we assume that for every smooth curve~$(\tilde{\rho}_\tau)_{\tau \in (-\delta, \delta)}$
in~$\calB$, the corresponding functions~$(f_\tau, F_\tau)$ have 
the properties that the derivatives
\begin{align}\label{Bassumption2}
 \frac{d}{d\tau} \L\big(F_\tau(x), F_\tau(y) \big)\: f_\tau(y) \qquad \textrm{and} \qquad  \frac{d}{d\tau} D_1 \L\big(F_\tau(x), F_\tau(y) \big)\: f_\tau(y)
\end{align}
are bounded uniformly in $\tau$ for every $x,y \in M$ and are~$\rho$-integrable in $y$ for every $x \in M$
(just as in~\eqref{IntreqLinFieldEq}, the subscripts of~$D_1$ and~$D_2$ denote the
partial derivatives acting on the first respectively second argument of the Lagrangian).

\begin{Lemma} \label{lemmalin}
For any~$\u \in \J$ and $\v \in T_\rho \calB$,
\beq \label{eqlin}
\nabla_\u \nabla_\v \ell(x) + \int_M \nabla_{1,\u} \nabla_{2,\v}\L(x,y)\: d\rho(y) = 0 \qquad \text{for all~$x \in M$}\:.
\eeq
\end{Lemma} \noindent
We refer to~\eqref{eqlin} as the {\em{linearized field equations}}
(see also the explanation in the introduction before~\eqref{IntreqLinFieldEq}).

Before giving the proof of this lemma, we point out that in~\eqref{eqlin}, the order of differentiation is irrelevant.
This is obvious for the term~$\nabla_{1,\u} \nabla_{2,\v}\L(x,y)$ because the derivatives act on
different variables. In the first term, it follows from the computation
\[  \nabla_\u \nabla_\v \ell(x) - \nabla_\v \nabla_\u \ell(x) = (D_u a)(x)\, \ell(x) - (D_v b)(x)\, \ell(x)
+ D_{[u,v]} \ell(x) = 0 \:, \]
where in the last step we used the weak EL equations~\eqref{ELweak}.
\Proof[Proof of Lemma~\ref{lemmalin}]
Given~$\v=(b,v) \in T_\rho \calB$, we let~$(\tilde \rho_\tau)_{\tau \in (-\delta,\delta)}$ be a smooth curve in~$\calB$
with~$\tilde{\rho}_0=\rho$ and~$\dot{\tilde{\rho}}_0 = \v$. 
As shown in Lemma~\ref{curve} below, there are~$F_\tau$ and~$f_\tau$ such that
\beq \label{paramtilrho}
\tilde{\rho}_\tau = (F_\tau)_* \big( f_\tau \,\rho \big) \:,
\eeq
and therefore
\beq \label{paramtilrho2}
\dot{\tilde{\rho}}_0 = \frac{d}{d\tau} \Big((F_\tau)_* \big( f_\tau \,\rho \big) \Big) \Big|_{\tau=0}
\qquad \text{with} \qquad \dot{f}_0 = b,\;\;\dot{F}_0 = v \:.
\eeq
Setting~$M_\tau = \supp \tilde{\rho}_\tau$ and using that~$M_\tau = \overline{F_\tau(M)}$,
the weak EL equations~\eqref{ELweak} can be written as
\beq \label{ELweak3}
\ell_\tau \big(F_\tau(x) \big) \equiv 0 \quad \text{and} \quad D \ell_\tau \big( F_\tau(x) \big) \equiv 0
\qquad \text{for all~$x \in M$}\:,
\eeq
where
\begin{align}\label{ellTauSmooth}
 \ell_\tau(z) := \int_\F \L(z,y)\: d\tilde{\rho}_\tau(y) - \frac{\nu}{2} \;\in\; C^\infty(\F, \R)\:.
\end{align}
Differentiating the first equation in~\eqref{ELweak3} with respect to~$\tau$, we obtain
\begin{align*}
0 &= \frac{d}{d\tau} \ell_\tau \big(F_\tau(x) \big) \big|_{\tau=0}
= \frac{d}{d\tau} \int_\F \L\big(F_\tau(x), y \big)\: d \Big( \big(F_\tau \big)_* \big( f_\tau \,\rho \big) \Big) (y)
\Big|_{\tau=0} \\
&= \frac{d}{d\tau} \int_M \L\big(F_\tau(x), F_\tau(y) \big)\: f_\tau(y)\: d\rho (y) \Big|_{\tau=0}\\
&= D_v \ell(x) + \int_M \L(x,y)\: b(y)\: d\rho(y) + \int_\F D_{2,v}\L(x,y)\: d\rho(y) \:.
\end{align*}
In the last step, we exchanged integration with differentiation. This is justified
by our assumption~\eqref{Bassumption2}, which 
ensures that the integrand of the second line is $L^1(\F,d\rho)$ for every $\tau$, is differentiable in $\tau$ for every $x,y \in M$ and is dominated by a $L^1(\F,\rho)$-function uniformly in $\tau$. 
Using the notation~\eqref{Djet}, we can write this as
\beq \label{jEL1}
D_v \ell(x) + \int_M \nabla_{2,\v}\L(x,y)\: d\rho(y) = 0 \:.
\eeq
Differentiating the second equation in~\eqref{ELweak3}, a similar computation gives
for any vector field~$u$
\beq \label{jEL2}
D_v D_u \ell(x) + \int_M D_{1,u} \nabla_{2,\v}\L(x,y)\: d\rho(y) = 0 \:.
\eeq
Multiplying~\eqref{jEL1} by~$a(x)$ and adding~\eqref{jEL2}, we obtain
\begin{align*}
0 &= a(x) \:D_v \ell(x) + D_v D_u \ell(x) + \int_M \nabla_{1,\u} \nabla_{2,\v}\L(x,y)\: d\rho(y) \\
&= D_v \nabla_\u \ell(x) - (D_v a)(x)\, \ell(x)+ \int_M \nabla_{1,\u} \nabla_{2,\v}\L(x,y)\: d\rho(y) \\
&= \nabla_v \nabla_\u \ell(x) - b(x)\, \nabla_\u \ell(x)
- (D_v a)(x)\, \ell(x)+ \int_M \nabla_{1,\u} \nabla_{2,\v}\L(x,y)\: d\rho(y) \:.
\end{align*}
Using the weak EL equations~\eqref{ELweak2}, the second and third summands vanish, giving the result.
\QED

\begin{Lemma}\label{curve}
Let $\rho \in \calB$ and $(\tilde \rho_\tau)_{\tau \in (-\delta,\delta)}$ be a curve such that $\tilde \rho_0 = \rho$.
Then there is a family of smooth diffeomorphisms~$F_\tau : \F \rightarrow \F$ and functions~$f_\tau \in C^\infty(\F, \R^+)$ such that
\begin{align}
\label{CurveCurve}
\tilde{\rho}_\tau = (F_\tau)_* \big( f_\tau \,\rho \big) \: .
\end{align}
If the curve~$(\tilde \rho_\tau)_{\tau \in (-\delta,\delta)}$ is smooth,
then both~$F_\tau$ and $f_\tau$ are smooth in $\tau$.
\end{Lemma}
\Proof Let $\rho_0$ be the measure in the definition of $\calB$,~\eqref{Bdef}. Then there are $G : \F \rightarrow \F$ and~$g \in C^\infty(\F, \R^+)$ such that
\[ \rho = G_* \big( g \,\rho_0 \big) \:.\]
Thus
\[ \rho_0 = \frac{1}{ g} \:\Big( \big(G^{-1} \big)_* \rho \Big) = ( G^{-1})_* \bigg( \frac{1}{ g \circ G^{-1} } \:\rho \bigg) \:. \]
Similarly, there are mappings~$\tilde G: \F \rightarrow \F$ and~$\tilde g\in C^\infty(\F, \R^+)$ such that
\begin{align*}
\tilde \rho_\tau &=  \tilde G_* \big( \tilde g \,\rho_0 \big) 
= \tilde G_* \bigg( \tilde g \:\big( G^{-1} \big)_* \Big( \frac{1}{ g \circ G^{-1} } \:\rho \Big)  \bigg) \\
&= (\tilde G \circ G^{-1})_* \bigg(  \frac{\tilde g \circ G^{-1} }{ g \circ G^{-1} }\: \rho \bigg) \:.
\end{align*}
This gives the desired functions~$F_\tau$ and~$f_\tau$ for fixed~$\tau$.
The claim about smoothness follows from the topology and the differential structure of $\calB$ as described Appendix~\ref{appclosed} (a curve~\eqref{CurveCurve} is smooth if and only if~$F_\tau$ and $f_\tau$ are smooth
in $\tau$).
\QED
We remark that the strict positivity of the weight function~$f$ in~\eqref{rhoFf} is needed
because in the last proof we divided by these weight functions.

\subsection{The Symplectic Form and Hamiltonian Time Evolution}\label{SecSympForm}
For any~$\u, \v \in T_\rho \calB$ and~$x,y \in M$, we set
\[ \sigma_{\u, \v}(x,y) := \nabla_{1,\u} \nabla_{2,\v} \L(x,y) - \nabla_{1,\v} \nabla_{2,\u} \L(x,y) \:. \]
For any compact~$\Omega \subset \F$, we introduce the surface layer integral
\beq \label{OSI}
\sigma_\Omega \::\: T_\rho \calB \times T_\rho \calB \rightarrow \R\:,
\qquad \sigma_\Omega(\u, \v) = \int_\Omega d\rho(x) \int_{M \setminus \Omega} d\rho(y)\:
\sigma_{\u, \v}(x,y) \:.
\eeq

We are now in the position to specify and prove the theorem mentioned in the introduction.
\begin{Thm} \label{thmOSI}
Let~$\F$ be a smooth manifold of dimension~$m \geq 1$, and~$\L \in C^\infty(\F \times \F, \R^+_0)$ 
be a smooth Lagrangian. Moreover, let~$\calB$ be a Fr{\'e}chet manifold of measures
of the form~\eqref{Bdef}.
Then for any compact~$\Omega \subset \F$, the surface layer integral~\eqref{OSI}
vanishes for all~$\u, \v \in T_\rho \calB$.
\end{Thm}
\Proof Anti-symmetrizing~\eqref{eqlin} in~$\u$ and~$\v$
and using that~$\nabla_{[\u, \v]} \ell = 0$, we obtain
\[ \int_M \sigma_{\u, \v}(x,y) \:d\rho(y) = 0 \qquad \text{for all~$x \in M$} \:. \]
We integrate this equation over~$\Omega$,
\begin{align}
0 &= \int_{\Omega} d\rho(x) \int_M  d\rho(y)\: \sigma_{\u, \v}(x,y) \notag \\
&= \int_{\Omega} d\rho(x) \int_\Omega d\rho(y)\: \sigma_{\u, \v}(x,y)
+ \int_{\Omega} d\rho(x) \int_{M \setminus \Omega} d\rho(y)\: \sigma_{\u, \v}(x,y) \:. \label{conserve}
\end{align}
Since the Lagrangian is symmetric in its two arguments, the
function~$\sigma_{\u, \v}$ is obviously anti-symmetric, i.e.\ $\sigma_{\u, \v}(x,y) = -\sigma_{\u, \v}(y,x)$.
Therefore, the first summand in~\eqref{conserve} vanishes. This gives the result.
\QED

At this point, we want to use the construction explained after~\eqref{IntrOSIN}
in the introduction to obtain a
conserved symplectic form~$\sigma$. In the present smooth setting, this construction can be
made precise as follows.
Let us assume that~$M$ is a smooth manifold being a topological product
\[ M = \R \times N \]
with a (possibly non-compact) smooth manifold~$N$. Then for any~$t \in \R$, the
set~$N_t := \{t\} \times N$ is a hypersurface in~$M$, and it can be realized as a boundary,
\[ N_t = \partial \Omega_{N_t} \qquad \text{with} \qquad \Omega_{N_t} := (-\infty, 0) \times N \:. \]
Next, let us assume that the jets~$\v = (b,v) \in T_\rho \calB$
have suitable decay properties at spatial infinity which ensure that the surface layer integrals~\eqref{IntrOSIN}
exist for~$\Omega_N=\Omega_{N_t}$ and every~$t \in \R$.
Under these assumptions, Theorem~\ref{thmOSI} implies that the bilinear form~$\sigma_{\Omega_{N_t}}$
is well-defined and does not depend on~$t$. This makes it possible to introduce the mapping
\beq \label{sigmaform}
\sigma \::\: T_\rho \calB \times T_\rho \calB \rightarrow \R \:, \qquad (\u, \v) \mapsto \sigma_{\Omega_{N_t}}(\u,\v)
\eeq
(where~$t \in \R$ is arbitrary). Due to the anti-symmetry, we can regard~$\sigma$ as a two-form
on~$\calB$. The next lemma shows that~$\sigma$ endows~$\calB$ with
the structure of a presymplectic Fr{\'e}chet manifold.
\begin{Lemma} \label{lemmaclosed}
The bilinear form~$\sigma$ is closed.
\end{Lemma}
\Proof Inspired by classical field theory (see for example~\cite[\S2.3]{deligne+freed}),
our strategy is to write~$\sigma$ locally as the exterior derivative of a one-form~$\gamma$.
Then the claim follows immediately from the fact that~$d^2=0$.

We let~$\tilde{\rho}$ be a measure in a neighborhood of~$\rho \in \calB$.
By definition of~$\calB$ we can represent~$\tilde{\rho}$ as
\beq \label{362}
\tilde{\rho} = F_* \big( f \,\rho \big) \in \calB \:.
\eeq
We next define~$\gamma : T_{\tilde{\rho}} \calB \rightarrow \R$ by
\beq
\gamma(\u) = \int_{\Omega_{N_t}} d\rho \int_{M \setminus \Omega_{N_t}} \!\! d\rho\:
f(x)\: \nabla_{2,\u} \L\big(F(x), F(y)\big)\: f(y) \:. \label{gammaUdef} 
\eeq
Computing the outer derivative with the formula
\[ (d\gamma)(\u,\v) = \u \gamma(\v) - \v \gamma(\u) - \gamma([\u,\v]) \:, \]
one finds that~$\sigma =  d\gamma$ (for details see Appendix~\ref{appclosed}). This concludes the proof.
\QED

In order to obtain a symplectic structure on~$\calB$, the presymplectic form~$\sigma$ must be non-degenerate.
We do not see a general reason why this should be the case.
Therefore, we proceed as follows. Given~$\rho \in \calB$,
an abstract method to obtain a non-degenerate form is to mod out the kernel of~$\sigma$
defined by
\[ \text{ker} \,\sigma = \{ \v \in T_\rho \calB \;\big|\; \sigma(\u, \v) = 0 \text{ for all~$\u \in T_\rho \calB$} \big\} \:. \]
In most applications, it is useful to choose concrete representatives
of the vectors of this quotient space. To this end, one chooses a maximal
subspace~$\J^\text{symp}$ of~$T_\rho \calB$ on which~$\sigma$ is non-degenerate
(the existence of such a subspace is guaranteed by Zorn's lemma).
Then the restriction
\[ \sigma \::\: \J^\text{symp} \times \J^\text{symp} \rightarrow \R \]
is non-degenerate. The specific choice of~$\J^\text{symp}$ depends on the application.

To summarize, the above constructions gave us a presymplectic form~$\sigma$
on~$T_\rho \calB$ which is given as a surface layer integral~\eqref{IntrOSI}
for~$\Omega=\Omega_{N_t}$. This presymplectic form is independent of~$t$.
In other words, the time evolution as specified by the linearized field equations~\eqref{eqlin}
preserves the symplectic form and is thus a symplectomorphism.
This is what we mean by {\em{Hamiltonian time evolution}}.

\section{The Lower Semi-Continuous Setting}\label{lowerSemiCont}
We now return to the lower semi-continuous setting of Section~\ref{seccvpcfs}.
We assume that~$\rho$ satisfies the EL equations of the causal action
(see~\eqref{elldef} and~\eqref{ELstrong}). Thus we assume that the function~$\ell : \F \rightarrow \R$ defined by
\beq
\ell(x) := \int_\F \L(x,y)\: d\rho(y) - \frac{\nu}{2} \quad \text{is bounded and lower semi-continuous,} \label{elldeflip}
\eeq
and that it is minimal on the support of~$\rho$,
\beq
\ell|_{\supp \rho} \equiv \inf_\F \ell = 0 \label{ELstronglip}
\eeq
(here~$\nu>0$ is again the Lagrange multiplier describing the volume constraint;
see~\cite[\S1.4.1]{cfs}). We again introduce space-time as the support of the
universal measure,
\[ M := \supp \rho\:. \]

\subsection{The Weak Euler-Lagrange Equations}\label{WeakElLSC}
Since the function~$\ell$ as defined in~\eqref{elldeflip} is only lower semi-continuous,
the derivative in~\eqref{ELweak} in general does not exist. But for lower semi-continuous
functions, it is a reasonable assumption that the semi-derivatives exist, but may
take the value $+\infty$. This leads us to the following additional assumptions:
\begin{itemize}[leftmargin=2.5em]
\item[(v)] $\L$ has directional semi-derivatives in~$\R \cup \{\infty \}$: For any~$x,y \in \F$, $v \in T_x \F$ and
any curve $\gamma \in C^1((-1,1), \F)$ with~$\gamma(0)=x$ and~$\gamma'(0)=v$, the following
generalized semi-derivative exists
\beq \label{DvL}
D^+_{1,v} \L(x,y) := \lim_{\tau \searrow 0} \frac{1}{\tau} \Big( \L \big( \gamma(\tau), y \big)
- \L \big( \gamma(0), y \big) \Big) \;\in\; \R \cup \{\infty\}
\eeq
and is independent of the choice of~$\gamma$.\label{Cond5}
\item[(vi)] For any~$x \in M$ and~$v \in T_x\F$, both sides of the 
following equation exist and are equal,\label{Cond6}
\[ D^+_{v} \ell(x) = \int_M D^+_{1,v} \L(x,y)\: d\rho(y) \:. \]
\end{itemize}
Under these assumptions, the EL equations~\eqref{ELstronglip} can again be tested weakly
with smooth jets. We define the jet space~$\J$ as in~\eqref{Jdef},
\beq \label{Jdeflip}
\J := \big\{ \u = (a,u) \text{ with } a \in C^\infty(\F, \R) \text{ and } u \in C^\infty(\F, T\F) \big\} \:.
\eeq
We remark that it would suffice to define the mappings~$f$ and~$F$ on an open neighborhood of~$M$. But,
keeping in mind that every such mapping can be extended smoothly to all of~$\F$, there is no
loss in generality to assume that~$f$ and~$F$ are defined on all of~$\F$.
When testing, only the restriction of the jets to~$M$ is of relevance. We thus define the jet space
\beq \label{JM}
\J|_M := \big\{ \u = (a,u) \text{ with } a \in C^\infty(M, \R) \text{ and } u \in C^\infty(M, T\F) \big\} \:,
\eeq
where smooth functions and sections on~$M$ are defined as those functions (respectively sections)
which have a smooth extension to~$\F$.
Since in general only the semi-derivatives exist, in contrast to~\eqref{ELweak2}
the weak EL equations read
\beq \label{ELweaklip}
\nabla^+_{\u} \ell(x) \geq 0 \qquad \text{for all~$x \in M$ and~$\u \in \J|_M$}\:,
\eeq
where, similar to~\eqref{Djet}, $\nabla^+_\u$ is defined as
\begin{align}\label{DjetSemi}
\nabla^+_{\u} \ell(x) := a(x)\, \ell(x) + \big(D^+_u \ell \big)(x) \:.
\end{align}

We introduce~$\Jdiff$ as the subspace of jets on~$M$ such that~$\ell$ is differentiable in
the direction of the vector field, i.e.
\[ \Jdiff := \{ \u \in \J|_M \text{ with } \nabla^+_{\u} \ell = - \nabla^+_{-\u} \ell \} \;\subset\; \J|_M\:. \]
Note that the last equation does not impose a condition for the scalar component of the jet, so that
\begin{align}\begin{split}
\Jdiff &= C^\infty(M, \R) \oplus \Gdiff \qquad \text{where} \\
\Gdiff \,&\!:= \{ u \in C^\infty(M, T\F) \text{ with } D^+_{u} \ell = - D^+_{-u} \ell \} \:.
\end{split} \label{JDiffLip}
\end{align}
Thus for jets in~$\Jdiff$, the directional derivatives exist, so that~$\nabla^+_u \ell = \nabla_u \ell$.
Then~\eqref{ELweaklip} implies that
\[ \nabla_{\u} \ell(x) = 0 \qquad \text{for all~$x \in M$ and~$\u \in \Jdiff$}\:. \]
As explained in the introduction after~\eqref{IntreqLinFieldEq}, in physical applications
it suffices to use only part of the information contained in these equations.
To this end, we choose a linear subspace
\beq\label{Gammatest}
\Jtest = \Ctest(M, \R) \oplus \Gtest \;\subset\; \Jdiff
\eeq
and consider the {\em{weak EL equations}}
\beq \label{ELtest}
\nabla_{\u} \ell|_M = 0 \qquad \text{for all~$\u \in \Jtest$}\:.
\eeq
The choice of~$\Jtest$ depends on the specific application and is of no relevance for
the remainder of this section.

\subsection{Families of Solutions and Linearized Solutions}\label{SmoothVar}
We now consider families of solutions of the weak EL equations.
To this end, we let~$(\tilde{\rho}_\tau)_{\tau \in (-\delta, \delta)}$ be a family of measures,
which similar to~\eqref{rhoFf} and~\eqref{362} we assume to be of the form
\beq \label{rhotau}
\tilde{\rho}_\tau = (F_\tau)_* \big( f_\tau \, \rho \big) \:,
\eeq
where~$f$ and~$F$ are smooth,
\[ f \in C^\infty\big((-\delta, \delta) \times \F \rightarrow \R^+ \big) \qquad \text{and} \qquad
F \in C^\infty\big((-\delta, \delta) \times \F \rightarrow \F \big) \:, \]
and have the properties~$f_0(x)=1$ and~$F_0(x) = x$ for all~$x \in M$.
%Moreover, we demand that the mapping
%\[ F_\tau |_M \::\: M \rightarrow \F \quad \text{is injective}\:. \]
Then the support of~$\tilde{\rho}_\tau$ is given by
\[ M_\tau := \supp \tilde{\rho}_\tau = \overline{F_\tau(M)} \:. \]
In order to formulate the weak EL equations~\eqref{ELtest} for~$\tau \neq 0$, there
is the complication that the jets in~$\Jtest$ are defined only on~$M$, whereas the
weak EL equations must be evaluated on~$M_\tau$. Therefore, we must introduce
a jet space~$\Jtest_\tau$ on~$M_\tau$. We choose~$\Jtest_\tau$ as the push-forward of~$\Jtest$
under~$F_\tau$, with an additional scalar component formed of the directional derivative of~$f_\tau$
(the reason for this choice will become clear in Lemma~\ref{lemmalinlip}).
More precisely,
\beq \label{Jtesttau}
\Jtest_\tau := \Big\{ \Big( (F_\tau)_* \big( a + D_u \log f_\tau \big), \:(F_\tau)_* u \Big) \text{ with } \u = (a,u) \in \Jtest \Big\} \:,
\eeq
where the push-forward is defined by
\beq \begin{split}
(F_\tau)_* a \:&:\: M_\tau \rightarrow \R^+ \:,\qquad \big( (F_\tau)_* a \big)(F_\tau(x)) = a(x) \\
(F_\tau)_* u \:&:\: M_\tau \rightarrow T\F \:,\qquad \big( (F_\tau)_* u \big)(F_\tau(x)) = DF_\tau|_x \,u(x)
\end{split} \label{pushforward}
\eeq
(equivalently, the last relation can be written as~$((F_\tau)_* u)|_{F_\tau(x)} \eta = u|_x (\eta \circ F_\tau)$
for any test function~$\eta$ defined in a neighborhood of~$F_\tau(x)$).

We point out that the push-forward of~$\Jdiff$ is in general not the same as the differentiable
jets corresponding to the measure~$\tilde{\rho}_\tau$, as is illustrated in the following example.

\begin{Example} (the causal variational principle on the sphere) \label{exsphere} {\em{
The causal variational principle on the sphere is obtained from the setting of causal
fermion systems by taking a mathematical simplification of a special case.
It was introduced in~\cite[Section~1]{continuum} (cf.~\cite[Examples~1.5, 1.6 and~2.8]{continuum})
and analyzed in more detail in~\cite[Section~5]{support}.
We choose~$\F = S^2$ and let~$\D \in C^\infty(\F \times \F, \R)$ be the smooth function
\[ \D(x,y) := 
2 \tau^2\: (1+ \langle x,y \rangle) \left( 2 - \tau^2 \:(1 - \langle x,y \rangle) \right) , \]
where $\tau \geq 1$ is a parameter of the model and $\langle .,. \rangle$ is the scalar product on~$\R^3$.
Obviously, the function~$\D$ depends only on the angle~$\vartheta \in [0,\pi]$ between the
points~$x, y \in S^2$ (defined by~$\cos \vartheta = \langle x, y \rangle$).
We here choose~$\tau = \sqrt{2}$, so that
\[ \D = \D(\vartheta) = 8 \: (1+ \cos \vartheta ) \cos \vartheta\:. \]
The function~$\D$ has a maximum at~$\vartheta=0$
and changes signs at~$\vartheta_{\max}:= \frac{\pi}{2}$; more precisely
\[ \D|_{[0, \vartheta_{\max})} > 0 \:, \quad \D(\vartheta_{\max})=0\:, \quad \D|_{(\vartheta_{\max}, \pi]} \leq 0 \:. \]
We define the Lipschitz-continuous {\em{Lagrangian}}~$\L$ by
\begin{align} \label{Lform}
\L = \max(0, \D) \in C^{0,1}(\F \times \F, \R^+_0) \:.
\end{align}
Hence $\L(\vartheta)$ is positive if and only if $0 \leq \vartheta < \vartheta_{\max}$.
Furthermore, $\L$ is not differentiable at $\vartheta = \vartheta_{\max}$ since the semi-derivatives $\partial_\vartheta^+ \L(\vartheta)$ and $\partial_\vartheta^- \L(\vartheta)$ do not agree at this point.

It is shown in~\cite{support} that, for our choice of $\tau$, a minimizer of the causal variational
principle~\eqref{Sact} is given by a normalized counting measure supported on an octahedron.
Thus, denoting the set of unit vectors in $\R^3$ by~$\mathbb B:= \{e_1,e_2,e_3\}$,
the measure
\begin{align*}
\rho = \frac{1}{6} \sum_{x \in  \pm \mathbb B } \delta_x
\end{align*}
is a minimizer (where~$\delta_x$ denotes the Dirac measure supported at~$x \in S^2$).
Note that for all distinct points~$x, y \in \supp \rho$, the angle~$\vartheta$
is either~$\frac{\pi}{2}$ or~$\pi$, implying that~$\L(x,y)=0$. As a consequence,
\begin{align*}
\ell(x) = \D(0) \qquad \textrm{ for all } x \in M \, .
\end{align*}

In order to determine $\Jdiff$ as defined in~\eqref{JDiffLip}, given any~$x \in M$ and a
non-zero vector~$u \in T_x\F$,
we let~$\gamma: (-\delta,\delta) \rightarrow \F$ be a smooth curve with~$\gamma(0) = x$
and~$\dot \gamma(0) = u$. Qualitatively speaking, the function $\ell(\gamma(\tau))$
has a ``cusp-like minimum'' at~$\tau=0$ because for $\tau > 0$, there is at least one point $y \in M$ which contributes to 
$\ell(\gamma(\tau))$ whereas for $\tau < 0$, the same is true for a different point $\tilde y$.
This can be made precise as follows.
There is at least one point~$y \in M$ with~$\vartheta_{x,y}=\vartheta_{\max}$
and~$\partial_\tau \vartheta_{\gamma(\tau), y}<0$ at $\tau=0$. This point
contributes to~$\ell$ for positive~$\tau$, i.e.
\[ \ell(\gamma(\tau)) \geq \D(\vartheta_{\gamma(\tau), x}) + \D(\vartheta_{\gamma(\tau), y})
\qquad \text{if~$\tau\geq0$}  \]
and thus
\begin{align*}
D^+_u \ell(x) &= \partial_\tau^+ \ell \big( \gamma(\tau) \big) \big|_{\tau=0} \geq
\partial_\tau \big( \D(\vartheta_{\gamma(\tau), x}) + \D(\vartheta_{\gamma(\tau), y}) \big) \big|_{\tau=0} \\
&=  \D'(0)\, \partial_\tau\vartheta_{\gamma(\tau), x}|_{\tau=0} + \D'(\vartheta_{\max})
\: \partial_\tau \vartheta_{\gamma(\tau), y} > 0 \:.
\end{align*}
Here, in the second step we used that $\D$ is differentiable, and hence the one-sided derivatives agree with the derivative, and in the last step we used that~$\D'(0)=0$ and~$\D'(\vartheta_{\max})<0$.
Likewise, there is a point~$\tilde{y} \in M$ with~$\vartheta_{x,\tilde{y}}=\vartheta_{\max}$
and~$\partial_\tau \vartheta_{\gamma(\tau), \tilde{y}}>0$. This point contributes to~$\ell$ for negative~$\tau$,
implying that
\begin{align*}
D_{-u} \ell(x) &= -\partial_\tau^- \ell \big( \gamma(\tau) \big) \big|_{\tau=0} \geq
-\partial_\tau \big( \D(\vartheta_{\gamma(\tau), x}) - \D(\vartheta_{\gamma(\tau), \tilde{y}}) \big) \big|_{\tau=0} \\
&= -\D'(0)\, \partial_\tau\vartheta_{\gamma(\tau), x}|_{\tau=0} - \D'(\vartheta_{\max})
\: \partial_\tau \vartheta_{\gamma(\tau), \tilde{y}} > 0 \:.
\end{align*}
Hence~$D^+_u \ell(x) \neq -D^+_{-u}\ell(x)$, so that the directional derivative~$D_u \ell(x)$
does not exist. We conclude that~$\Gdiff = \{0\}$ and thus
\beq \label{Jdiffex}
\Jdiff = C^\infty(M) \oplus \{0\} \:.
\eeq

We next define a family of measures~$\tilde{\rho}_\tau$ as in~\eqref{rhotau}.
To this end, we choose trivial weight functions~$f_\tau \equiv 1$, and choose
a family of diffeomorphism which change the angles between the points of~$M$ in the sense that
\[ \vartheta_{F_\tau(x), F_\tau(y)} \neq \vartheta_{\max} \qquad \text{for all~$x,y \in M$ and all~$\tau \neq 0$}\:. \]
Then for any~$\tau \neq 0$, the Lagrangian is smooth on a neighborhood of~$M_\tau \times M_\tau$, so that
\[ \Jdiff(\tilde{\rho}_\tau) = C^\infty(M_\tau) \oplus C^\infty(M_\tau, T\F) \:. \]
On the other hand, using the formula in~\eqref{Jtesttau} to define the push-forward of~$\Jdiff$
as given by~\eqref{Jdiffex}, we obtain the jet space~$C^\infty(M_\tau) \oplus \{0\}$.
Hence the push-forward of~$\Jdiff$ does not coincide with the 
differentiable jets of the measure~$\tilde{\rho}_\tau$.
}} \QEDrem
\end{Example}

In the next lemma we bring the weak EL equations for families of solutions into a convenient form.
To this end, as in~\eqref{ellTauSmooth}, we define
\[ \ell_\tau(z) = \int_\F \L(z,y)\: d{\rho}_\tau(y) - \frac{\nu}{2} \:. \]

\begin{Lemma} \label{lemmalinlip}
Assume that~$(\tilde{\rho}_\tau)_{\tau \in (-\delta, \delta)}$ is a family of
solutions of the weak EL equations~\eqref{ELtest} in the sense that
\beq\label{weakNull}
\nabla_{\u(\tau)} \ell_\tau(z) = 0 \qquad \text{for all~$z \in M_\tau$ and~$\u(\tau) \in \Jtest_\tau$}\:.
\eeq
Then for any~$\u \in \Jtest$ and all~$\tau \in (-\delta, \delta)$,
\beq \label{final}
0 = \nabla_{\u} \bigg( \int_M f_\tau(x) \:\L\big(F_\tau(x), F_\tau(y) \big)\: f_\tau(y)\: d\rho(y) 
-\frac{\nu}{2} \: f_\tau(x) \bigg) \:.
\eeq
\end{Lemma} \noindent
\Proof Writing the jet~$\u(\tau)$ as in~\eqref{Jtesttau} as
\[ \u(\tau) = \Big( (F_\tau)_* \big( a + D_u \log f_\tau \big), \:(F_\tau)_* u \Big) \]
with~$\u \in \Jtest$ and using the definition of the
push-forward~\eqref{pushforward}, the weak EL equations~\eqref{weakNull} yield
\beq \label{prelim}
0 = \nabla_{\tilde{\u}} \bigg( \int_M \L\big(F_\tau(x), F_\tau(y) \big)\: f_\tau(y)\: d\rho(y) - \frac{\nu}{2} \bigg)
\eeq
with~$\tilde{\u}=(a+(D_u \log f_\tau), \,u)$, valid for all~$\tau \in (-\delta, \delta)$ and $x \in M$.
Multiplying by~$f_\tau(x)$, we obtain
\begin{align*}
0 &= f_\tau(x) \:\nabla_{\tilde{\u}} \bigg( \int_M \L\big(F_\tau(x), F_\tau(y) \big)\: f_\tau(y)\: d\rho(y) -\frac{\nu}{2} \bigg) \\
&= f_\tau(x) \Big( a(x) + (D_u \log f_\tau)(x) + D_u \Big)
\bigg(  \int_M \L\big(F_\tau(x), F_\tau(y) \big)\: f_\tau(y)\: d\rho(y) -\frac{\nu}{2} \bigg) \\
&= \big( a(x) + D_u \big) \: f_\tau(x) \bigg(  \int_M \L\big(F_\tau(x), F_\tau(y) \big)\: f_\tau(y)\: d\rho(y) 
-\frac{\nu}{2} \bigg) \\
&= \nabla_\u \bigg( \int_M f_\tau(x) \:\L\big(F_\tau(x), F_\tau(y) \big)\: f_\tau(y)\: d\rho(y) -\frac{\nu}{2} \: f_\tau(x) \bigg)\:,
\end{align*}
making it possible to write~\eqref{prelim} ``more symmetrically'' in the form~\eqref{final}.
\QED

Differentiating~\eqref{final} naively with respect to~$\tau$, we obtain
the {\em{linearized field equations}}
\beq \label{eqlinlip}
\la \u, \Delta \v \ra|_M = 0
\eeq
with
\beq \label{eqlinlip2}
\la \u, \Delta \v \ra(x) := \nabla_{\u} \bigg( \int_M \big( \nabla_{1, \v} + \nabla_{2, \v} \big) \L(x,y)\: d\rho(y) - \nabla_\v \:\frac{\nu}{2} \bigg) \:,
\eeq
where~$\v$ is the jet~$\v = (\dot{f}_0, \dot{F}_0)$.
In order to see the connection to the linearized field equations in the smooth setting~\eqref{eqlin},
we note that a formal computation using~\eqref{elldeflip} gives
\[ \la \u, \Delta \v \ra(x) = \nabla_\u \nabla_\v \ell(x) + \int_M \nabla_{1, \u} \nabla_{2, \v} \L(x,y)\: d\rho(y)\:. \]
In the present lower semi-continuous setting, it is not clear if the derivatives exist, nor if
the derivatives may be interchanged with the integrals. It turns out to be preferable to
work with~\eqref{eqlinlip2}. In order to make sense of this expression, we
need to impose conditions on~$\v$. This leads us to the following definition:

\begin{Def} \label{deflin} A jet~$\v \in \J$ 
is referred to as a {\bf{solution of the linearized field equations}}
(or for brevity a {\bf{linearized solution}}) if it has the following properties:
\begin{itemize}[leftmargin=3em]
\item[\rm{(l1)}] For all~$y \in M$ and all~$x$ in an open neighborhood of~$M$, the 
following combination of derivatives exists,
\beq \label{derex1}
\big( \nabla_{1, \v} + \nabla_{2, \v} \big) \L(x,y) \;\in\; \R \:.
\eeq
Here the combination of directional derivatives is defined by
\[ \big( D_{1, v} + D_{2, v} \big) \L(x,y) := \frac{d}{d\tau} 
\L\big( F_\tau(x), F_\tau(y) \big) \big|_{\tau=0} \:,\]
where~$F_\tau$ is the flow of the vector field~$v$.
\item[\rm{(l2)}] Integrating the expression~\eqref{derex1} over~$y$ with respect to the measure~$\rho$,
the resulting function (defined on an open neighborhood of~$M$) is differentiable in the
direction of every jet~$\u \in \Jtest$ and satisfies the linearized field equations~\eqref{eqlinlip}.
\end{itemize}
The vector space of all linearized solutions is denoted by~$\Jlin \subset \J$.
\end{Def}

In order to illustrate the significance of the condition~\eqref{derex1}, we now give an
an example where the derivative in~\eqref{derex1} does not exist.

\begin{Example} (the causal variational principle on the sphere continued) \label{exspherecont} {\em{
We return to the causal variational principle on the sphere as considered in Example~\ref{exsphere}.
We first want to give an example where for $\L$ as given in~\eqref{Lform},
the condition~\eqref{derex1} is violated.
To this end, we choose two points~$x, y \in S^2$ such that~$\vartheta_{x, y} = \vartheta_{\max}$.
Furthermore, we choose a vector field~$v$ which vanishes in a neighborhood of~$y$. Then
\beq \label{LFF}
\L(F_\tau(x),F_\tau(y)) = \L(F_\tau(x),y)
\eeq
(where~$F_\tau$ is again the flow of the vector field~$v$).
Next, we choose~$v(x)$ to be nonzero, tangential to the great circle joining~$x$ and~$y$
and pointing in the direction of smaller geodesic distance to $y$. Then
\[ D^+_{1, -v} \L(x,y) = 0 \qquad \text{but} \qquad
D^+_{1, v} \L(x,y) = -\D'(\vartheta_{\max}) > 0\:. \]
Hence~\eqref{LFF} is not differentiable at~$\tau=0$.
We conclude that the derivative in~\eqref{derex1} does not exist.

Despite the just-explained difficulty to satisfy the condition~\eqref{derex1},
the space~$\Jlin$ contains non-trivial vector fields. For example, if~$(F_\tau)_{\tau \in (-\delta, \delta)}$
is a smooth family of isometries of the sphere (for example rotations around a fixed axis), then the
corresponding jet~$\v = (0, v)$ with $v = \partial_\tau F_\tau|_{\tau=0}$ is in~$\Jlin$.
This follows because
\begin{align*}
\frac{d}{d\tau} \L\big( F_\tau(x), F_\tau(y) \big) \big|_{\tau=0} = 0 \, ,
\end{align*}
i.e.\ the derivatives in~\eqref{derex1} exist and are zero. It follows that condition (l2),
including the linearized field equations, are satisfied as well.
This shows that~$\v \in \Jlin$.
We conclude that~$\Jlin$ is a vector space of dimension at least three.

In view of~\eqref{Jdiffex}, this example also illustrates that~$\Jlin$ is in general not a subspace of~$\Jdiff$.
}} \QEDrem
\end{Example}

\subsection{The Symplectic Form and Hamiltonian Time Evolution}
Following the construction in Section~\ref{SecSympForm}, we want
to anti-symmetrize the linearized field equations~\eqref{eqlinlip} in the jets~$\u$ and~$\v$.
To this end, we now consider~$\u, \v \in \Jtest \cap \Jlin$.
The conditions in Definition~\ref{deflin} ensure that the linearized field equations~\eqref{eqlinlip} are well-defined.
For the construction of the symplectic form, we need additional technical assumptions.
In order to make minimal assumptions, we work with semi-derivatives only. The existence of right semi-derivatives in $\R \cup  \{\infty\}$ is
guaranteed by condition~(v) on page~\pageref{DvL}. We define the left semi-derivative as
\[ D^-_{1,v} \L(x,y) := \lim_{\tau \nearrow 0} \frac{1}{\tau} \Big( \L \big( \gamma(\tau), y \big)
- \L \big( \gamma(0), y \big) \Big) \;\in\; \R \cup \{\infty\} \:, \]
so that~$D^-_{1,v} = - D^+_{1,-v}$. Similar to~\eqref{DjetSemi}, we define
\[ \nabla^-_{\u} \ell(x) := a(x)\, \ell(x) + \big(D^-_u \ell \big)(x) \:. \]
Before stating the additional assumptions, we explain why they are needed. First, we must
ensure that the individual terms in~\eqref{eqlinlip} exist and that we may exchange the differentiation with integration.
Second, when taking second derivatives, we must take into account that the jets are also differentiated.
We use the notation
\[ \nabla^s_{\u(x)} \nabla^{s'}_{1, \v} \L(x,y) \]
to indicate that the $\u$-derivative also acts on~$\v$. 
With this notation, we can state the additional technical assumptions as follows:
\begin{itemize}[leftmargin=3em]
\item[\rm{(s1)}] The first and second semi-derivatives of the Lagrangian
in the direction of~$\Jtest \cap \Jlin$ exist in~$\R$. \label{CondS1}
Moreover, for all~$x$ and~$y$ in a neighborhood of~$M$, the symmetrized first semi-derivatives of the Lagrangian
\[ \big( \nabla^+_{1,\u} + \nabla^-_{1,\u} \big) \L(x,y) \]
are linear in~$\u \in \Jtest \cap \Jlin$.
\item[\rm{(s2)}] The second semi-derivatives can be interchanged with the $M$-integration, i.e.\
for all~$\u, \v \in \Jtest$ and~$s, s' \in \{\pm\}$,
\begin{align*}
\int_M \nabla^s_{\u(x)} \nabla^{s'}_{1,\v} \L(x,y) \:d\rho(y) &= 
\nabla^s_{\u} \int_M \nabla^{s'}_{1,\v}  \L(x,y) \:d\rho(y)
%= \nabla^s_{\u(x)} \nabla^{s'}_{1,\v}  \int_M \L(x,y) \:d\rho(y)
\\
\int_M \nabla^s_{1,\u} \nabla^{s'}_{2,\v} \L(x,y) \:d\rho(y) &= \nabla^s_{\u} \int_M \nabla^{s'}_{2,\v} \L(x,y) \:d\rho(y) \:.
\end{align*}
\item[\rm{(s3)}] For any~$\u, \v \in \Jtest \cap \Jlin$, the 
commutator~$[\u, \v]$ is in~$\Jtest$.
\end{itemize}

\begin{Thm} \label{thmOSIlip}
Under the above assumptions~{\em{(l1), (l2)}} and~{\em{(s1)--(s3)}}, 
for any compact~$\Omega \subset \F$, the surface layer integral
\beq \label{OSIlip}
\sigma^{s,s'}_\Omega(\u, \v) = \int_\Omega d\rho(x) \int_{M \setminus \Omega} d\rho(y)\:
\sigma^{s,s'}_{\u, \v}(x,y) 
\eeq
with
\beq \label{OSIIntegrand}
\sigma^{s,s'}_{\u, \v}(x,y) := \nabla^s_{1,\u} \nabla^{s'}_{2,\v} \L(x,y) - \nabla^{s'}_{1,\v} \nabla^s_{2,\u} \L(x,y)
\eeq
vanishes for all~$\u, \v \in \Jtest \cap \Jlin$ and all~$s,s' \in \{\pm\}$.
\end{Thm}
\Proof Our starting point are the linearized field equations~\eqref{eqlinlip}.
Condition (s1) ensures that we can treat the terms of~\eqref{eqlinlip2} independently if we
take semi-derivatives. First, 
for $\u, \v \in \Jlin \cap \Jtest$
we consider the term
\[ \nabla^s_{\u} \int_M \nabla^{s'}_{1, \v} \L(x,y)\: d\rho(y)  =  \int_M \nabla^s_{\u(x)} \nabla^{s'}_{1, \v} \L(x,y)\:
d\rho(y) \:. \]
We now exchange $\u$ and $\v$ as well as $s$ and $s'$ and take the difference. Using the relation
\[ \nabla_{1,\v}^{s} = ss' \, \nabla^{s'}_{1, ss' \v} \:, \]
we obtain
\begin{align*}
&\int_M \big( \nabla^s_{\u(x)} \nabla^{s'}_{1, \v} -\nabla^{s'}_{\v(x)} \nabla^{s}_{1, \u} \big)\L(x,y)\: d\rho(y) =
s s' \,  \int_M  \nabla^s_{1,ss' [\u,\v]} \L(x,y)\: d\rho(y)\\
&\;\;= s s' \,  \nabla^s_{ss' [\u,\v]} \Big( \ell(x) + \frac{\nu}{2} \Big)  =   \nabla_{[\u,\v]} \Big( \ell(x) + \frac{\nu}{2} \Big)  = \nabla_{[\u,\v]}  \frac{\nu}{2}
\end{align*}
Here, in the second step we used the assumption~(vi). In the third step, we applied assumptions~(s3)
and~\eqref{Gammatest}. In the last step, we used the weak EL equations.

Using this result, anti-symmetrizing~\eqref{eqlinlip2} (again by exchanging~$\u$ and $\v$ as well as $s$ and $s'$)
and using~\eqref{eqlinlip}, we obtain the equations
\[ \int_M \sigma^{s,s'}_{\u, \v}(x,y) d\rho(y) = 0 \:. \]
Integrating over~$\Omega$ gives
\beq \label{term1.2}
\int_\Omega d\rho(x) \int_M d\rho(y)\: \sigma^{s,s'}_{\u, \v}(x,y) = 0 \:.
\eeq
Moreover, it is obvious by the anti-symmetry of~\eqref{OSIIntegrand} that
\[ \int_\Omega d\rho(x) \int_\Omega d\rho(y)\: \sigma^{s,s'}_{\u, \v}(x,y) = 0 \:. \]
Subtracting this equation from~\eqref{term1.2} gives the result.
\QED

Theorem~\ref{thmOSIlip} again makes it possible to introduce a Hamiltonian time evolution,
just as explained in the introduction and in Section~\ref{SecSympForm}.
The only difference compared to Section~\ref{SecSympForm}
is that, due to the semi-derivatives in~\eqref{OSIIntegrand}, the surface layer integral~\eqref{OSIlip}
is in general not linear in the jets~$\u$ and~$\v$. In order to obtain a bilinear form,
we modify~\eqref{sigmaform} as follows.
\begin{Prp} The mapping~$\sigma$ defined by
\[ \sigma \::\: (\Jtest \cap \Jlin) \times (\Jtest \cap \Jlin) \rightarrow \R \:, \qquad (\u, \v) \mapsto
\frac{1}{4} \sum_{s,s'=\pm} \sigma^{s,s'}_{\Omega_{N_t}}( \u, \v) \]
is bilinear.
\end{Prp}
\Proof We have
\begin{align*}
\sigma(\u,\v) = \int_{\Omega_{N_t}} d\rho(x) \int_{M \setminus \Omega_{N_t}} d\rho(y)\:  \sigma_{\u, \v}(x,y) 
\end{align*}
with
\begin{align*}
&\sigma_{\u, \v}(x,y) =  \frac{1}{4} \sum_{s,s'=\pm} \sigma^{s,s'}_{\u, \v}(x,y) \\
&\;= \frac{1}{4} \:  \Big( \nabla_{1,\u}^+ + \nabla_{1,\u}^- \Big) \Big( \nabla_{2,\v}^+  + \nabla_{2,\v}^- \Big) \L(x,y)
- \frac{1}{4} \:  \Big( \nabla_{1,\v}^+ + \nabla_{1,\v}^- \Big) \Big( \nabla_{2,\u}^+ + \nabla_{2,\u}^- \Big) \L(x,y) \,.
\end{align*}
The assumption (s1) and the symmetry of $\L(x,y)$ (condition (i) on page~\pageref{Cond1}) imply that $\big( \nabla^+_{i,\u} + \nabla^-_{i,\u} \big) \L(x,y)$
is linear in $\u \in \Jtest \cap \Jlin$ for $i=1,2$. Hence~$\sigma_{\u, \v}(x,y)$ is indeed linear in~$\u$
and~$\v$.
\QED
Since the other considerations at the end of Section~\ref{SecSympForm} apply without changes,
we do not repeat them here.

\subsection{Local Minimizers and Second Variations} \label{seclocmin}
We now introduce the concept of local minimizers of causal variational principles
and explore the connection to second variations.
We derive a convenient criterion which ensures that a measure~$\rho$ is
a local minimizer (Proposition~\ref{LocSufficient}).
This criterion will be used in the example of Section~\ref{seclattice}
to prove the existence of local minimizers.

We again consider families of variations~$(\tilde{\rho}_\tau)_{\tau \in [0, \delta)}$
with~$\tilde{\rho}_0=\rho$. We assume that the measures~$\tilde{\rho}_\tau$ all satisfy
the conditions in~\eqref{totvol}. Then the family of measures~$\mu_\tau$ defined by
\beq \label{mudef}
\mu_\tau := \tilde{\rho}_\tau - \rho \:,
\eeq
are in the Banach space~${\mathfrak{B}}(\F)$ of signed measures on~$\F$
with the norm given by the total variation.
\begin{Def} \label{deflocmin}
The measure~$\rho$ is a {\bf{local minimizer}} of the causal action if for
every family~$(\tilde{\rho}_\tau)_{\tau \in [0,\delta)}$ of Borel measures which has the property that~$\mu_\tau$
defined by~\eqref{mudef} is a smooth regular
curve~$\mu : [0, \delta) \rightarrow {\mathfrak{B(\F)}}$ with $\mu_\tau(\F) = 0$,
there is~$\delta_0 \in (0, \delta)$ such that
\beq \label{Srhoin}
\big( \Sact(\tilde{\rho}_\tau) - \Sact(\rho) \big) \geq 0 \qquad \text{for all~$\tau \in [0, \delta_0)$} \:.
\eeq
\end{Def}

We first derive the implications of local minimality. To this end, we assume that~$\rho$ is a local minimizer.
Then obviously the EL equations~\eqref{EL1} hold, because the curve~\eqref{tilderho}
has the properties in the above definition. We consider
variations~$(\tilde{\rho}_\tau)_{\tau \in (-\delta, \delta)}$ of the form
\beq \label{rhovarpsi}
\tilde{\rho}_\tau = (1 + \tau \psi)\: \rho\:,
\eeq
where~$\psi$ is a real-valued function on~$\F$.
In order to ensure that these measures are again positive for sufficiently small~$\delta>0$,
we must assume that~$\psi$ is
an essentially bounded function. Moreover, these measures satisfy the conditions in the above
definition if and only if
\[ \int_M |\psi|\: d\rho < \infty \qquad \text{and} \qquad \int_M \psi\: d\rho = 0 \:. \]
Hence we must assume that~$\psi$ is in the space
\beq \label{DomainLrho}
{\mathscr{D}}(\L_\rho) := \Big\{ \psi \in (L^1 \cap L^\infty)(M, d\rho) \:\Big|\: \int_M \psi\: d\rho = 0 \Big\} \:.
\eeq
Then by interpolation, the function~$\psi$ is also a vector in the Hilbert space~$L^2(M, d\rho)$,
also denoted by~$(\H_\rho, \la .,. \ra_\rho)$.
The operator~$\L_\rho$ was already analyzed in~\cite[Lemma~3.5]{support} in the compact setting.
We now extend this analysis to the non-compact setting.

\begin{Lemma} For~$\psi \in (L^1 \cap L^\infty)(M, d\rho)$, the function~$\L_\rho \psi$ defined by
\[ (\L_\rho \psi)(x) = \int_M \L(x,y)\: \psi(y)\: d\rho(y) \]
is in~$L^2(M, d\rho)$, giving rise to a linear operator
\[ \L_\rho \::\: {\mathscr{D}}(\L_\rho)  \subset \H_\rho \rightarrow \H_\rho \:. \]
\end{Lemma}
\Proof We apply Tonelli's theorem to obtain
\begin{align*}
\int_M \big| \L_\rho \psi \big|^2\: d\rho
&\leq \|\psi\|_{L^\infty(M)} \int_M d\rho(x) \int_M d\rho(y)\: \L(x,y)\: |\psi(y)|  \int_M d\rho(y')\: \L(x,y') \\
&= \|\psi\|_{L^\infty(M)} \int_M d\rho(x) \int_M d\rho(y)\: \L(x,y)\: |\psi(y)| \:\Big( \ell(x) + \frac{\nu}{2} \Big) \\
&\leq \|\psi\|_{L^\infty(M)}\: \sup_{M} \Big( \ell + \frac{\nu}{2} \Big) \int_M d\rho(x) \int_M d\rho(y)\: \L(x,y) \: |\psi(y)| \\
&= \|\psi\|_{L^\infty(M)}\: \sup_{M} \Big( \ell + \frac{\nu}{2} \Big) \int_M d\rho(y)\: |\psi(y)| \int_M d\rho(x) \: \L(x,y) \\
&\leq \, \|\psi\|_{L^\infty(M)}\: \sup_{M} \Big( \ell + \frac{\nu}{2} \Big)^2\: \|\psi\|_{L^1(M)}
< \infty \:,
\end{align*}
where we used that~$\L$ is symmetric and that~$\ell$ is bounded according to our assumption~\eqref{elldef}.
This gives the result.
\QED

\begin{Prp} If~$\rho$ is a local minimizer, then the operator~$\L_\rho
: {\mathscr{D}}(\L_\rho) \rightarrow \H_\rho$ is positive (but not necessarily strictly positive).
\end{Prp}
\Proof Computing~\eqref{integrals2} for the variation~\eqref{rhovarpsi} gives
\beq \begin{split}
\big( \Sact(\tilde{\rho}) - \Sact(\rho) \big) &= 2 \tau \int_\F \Big(\ell(x) + \frac{\nu}{2} \Big) \:\psi(x)\: d\rho \\
&\qquad + \tau^2 \int_\F \psi(x)\: d\rho(x) \int_\F \psi(y)\: d\rho(y)\: \L(x,y) \:.
\end{split} \label{216}
\eeq
The first summand vanishes in view of the EL equations~\eqref{EL1}. The second summand,
on the other hand, exists in view of the estimates
\begin{align*}
\int_\F &\psi(x)\: d\rho(x) \int_\F \psi(y)\: d\rho(y)\: \L(x,y)
\leq \|\psi\|_{L^\infty(M)}\: \int_\F \psi(x)\: d\rho(x) \int_\F  d\rho(y)\: \L(x,y) \\
& \leq \|\psi\|_{L^\infty(M)}\: \sup_{M} \Big( \ell + \frac{\nu}{2} \Big) \int_\F \psi(x)\: d\rho(x)  = \|\psi\|_{L^\infty(M)}\: \sup_{M} \Big( \ell + \frac{\nu}{2} \Big) \|\psi\|_{L^1(M)} \, .
\end{align*}
Rewriting the second summand in~\eqref{216} as an expectation value, we obtain
\[ \big( \Sact(\tilde{\rho}) - \Sact(\rho) \big) = \tau^2 \: \la \psi, \L_\rho \psi \ra_\rho \:. \]
Applying the inequality~\eqref{Srhoin} gives the result.
\QED

We finally give a criterion which ensures that~$\rho$ is a local minimizer.

\begin{Prp}\label{LocSufficient} Let~$\rho$ be a Borel measure with the following properties:
\begin{itemize}[leftmargin=2em]
\item[{\rm{(a)}}] The EL equations~\eqref{ELstrong} are satisfied and in addition
\beq \label{imply}
\ell(x) = 0 \quad \Longrightarrow \quad x \in \supp \rho\:.
\eeq
\item[{\rm{(b)}}] The Lagrangian~$\L : \F \times \F \rightarrow \R^+_0$ is
a bounded function.
\item[{\rm{(c)}}] The operator~$\L_\rho
: {\mathscr{D}}(\L_\rho) \rightarrow \H_\rho$ is strictly positive in the sense that
there is~$\varepsilon>0$ such that
\beq \label{spos}
\la \psi, \L_\rho \psi \ra_\rho \geq \varepsilon\: \|\psi\|_\rho^2 \qquad \text{for all~$\psi \in {\mathscr{D}}(\L_\rho)$} \:.
\eeq
\end{itemize}
Then~$\rho$ is a local minimizer.
\end{Prp} \noindent
We remark that condition~(b) could be replaced by weaker boundedness assumptions.
We do not aim for maximal generality because condition~(b) is suitable for
the applications we have in mind.

\Proof[Proof of Proposition~\ref{LocSufficient}]
Let~$(\tilde{\rho}_\tau)_{\tau \in [0, \delta)}$ be as in Definition~\ref{deflocmin}.
Since the curve~$\mu_\tau$ in Definition~\ref{deflocmin} is regular, we know
that~$\dot{\tilde{\rho}}_0$ is non-zero.
Expanding~\eqref{integrals2} in powers of~$\tau$, we obtain
\begin{align*}
\big( \Sact(\tilde{\rho}) - \Sact(\rho) \big) &= 2 \tau \int_\F \Big(\ell(x) + \frac{\nu}{2} \Big) \:d\dot{\tilde{\rho}}_0(x) 
+ \tau^2 \int_\F \Big(\ell(x) + \frac{\nu}{2} \Big) \:d\ddot{\tilde{\rho}}_0(x) \\
&\quad + 2 \tau^2 \int_\F d\dot{\tilde{\rho}}_0(x) \int_\F d\dot{\tilde{\rho}}_0(x)\: \L(x,y) + \O\big(\tau^3 \big)\:.
\end{align*}
Due to the volume constraint, the signed measures~$\dot{\tilde{\rho}}_0$ and~$\ddot{\tilde{\rho}}_0$
have total volume zero, so that the terms involving~$\nu$ drop out,
\beq \label{vary}
\begin{split}
\big( \Sact(\tilde{\rho}) - \Sact(\rho) \big) &= 2 \tau \int_\F \ell(x) \:d\dot{\tilde{\rho}}_0(x) 
+ \tau^2 \int_\F \ell(x)\:d\ddot{\tilde{\rho}}_0(x) \\
&\quad + 2 \tau^2 \int_\F d\dot{\tilde{\rho}}_0(x) \int_\F d\dot{\tilde{\rho}}_0(x)\: \L(x,y) + \O\big(\tau^3 \big)\:.
\end{split}
\eeq

We first consider the case that the measure~$\chi_{\F \setminus M} \dot{\tilde{\rho}}_0$ is non-zero.
Since the measures~$\tilde{\rho}_\tau$ are all positive,
we know that~$\chi_{\F \setminus M} \dot{\tilde{\rho}}_0$ is a positive measure.
Hence, using~\eqref{imply}, we conclude that
\[ \int_\F \ell(x)\:d\dot{\tilde{\rho}}_0(x) > 0 \:. \]
Hence the linear term in~\eqref{vary} ensures that~\eqref{Srhoin} holds for sufficiently small~$\tau$.

It remains to consider the case that the measure~$\dot{\tilde{\rho}}_0$ is supported on~$M$.
Then the linear term in~\eqref{vary} vanishes because of the EL equations~\eqref{EL1}.
Repeating the above argument with~$\dot{\tilde{\rho}}_0$ replaced by~$\ddot{\tilde{\rho}}_0$, we
find that $\chi_{\F \setminus M} \ddot{\tilde{\rho}}_0$  is a positive measure.
From this it follows that
\[ \int_\F \ell(x)\:d\ddot{\tilde{\rho}}_0(x) \geq 0 \:. \]
Therefore, in order to conclude the proof, it remains to show that
\begin{align}\label{show3}
\int_M d\dot{\tilde{\rho}}_0(x) \int_M d\dot{\tilde{\rho}}_0(y)\: \L(x,y) > 0 \:.
\end{align}

We now use the following approximation argument. We choose a sequence~$\psi_n \in
{\mathscr{D}}(\L_\rho)$ such that
\beq \label{measurelim}
\psi_n\: \rho \rightarrow \dot{\tilde{\rho}}_0 \neq 0 \qquad \text{in~$\mathfrak{B(\F)}$}\:.
\eeq
Then
\begin{align*}
\int_M \psi_n(x)\: d\rho(x) \int_M \psi_n(y)\: d\rho(y)\: \L(x,y)  \rightarrow \int_M d\dot{\tilde{\rho}}_0(x) \int_M d\dot{\tilde{\rho}}_0(x)\: \L(x,y),
\end{align*}
because, setting~$\dot{\tilde{\rho}}_0 = \psi_n\: \rho + \Delta\rho$, we have
\begin{align*}
&\int_M d\dot{\tilde{\rho}}_0(x) \int_M d\dot{\tilde{\rho}}_0(x)\: \L(x,y) - \int_M \psi_n(x)\: d\rho(x) \int_M \psi_n(y)\: d\rho(y)\: \L(x,y)\\
&= \int_M d\Delta\rho(x) \int_M d\Delta\rho(y) \, \L(x,y) + 2 \int_M d\Delta\rho(x) \int_M \psi_n(y)\: d\rho(y) \, \L(x,y) \\
&\leq C \, \| \Delta \rho \|_{\mathfrak{B}(\F)}^2 + 2 \, \| \psi_n \|_{L^\infty(M)} \, \sup_{M} \Big( \ell + \frac{\nu}{2} \Big) \, \| \Delta \rho \|_{\mathfrak{B}(\F)} \rightarrow 0 \:,
\end{align*}
where $C:=\sup_{x,y \in \F}\L(x,y)$ is the pointwise bound of the Lagrangian.
Using the strict positivity~\eqref{spos}, we have
\begin{align} \label{sposuse}
\varepsilon\: \|\psi_n\|_\rho^2 \leq
\la \psi_n, \L_\rho \psi_n \ra_\H &= \int_M \psi_n(x)\: d\rho(x) \int_M \psi_n(y)\: d\rho(y)\: \L(x,y) ,
\end{align}
hence the left hand side of~\eqref{show3} cannot be negative. 
Let us assume that it is zero,
\begin{align}\label{limNull}
\int_M d\dot{\tilde{\rho}}_0(x) \int_M d\dot{\tilde{\rho}}_0(y)\: \L(x,y) = 0 \, .
\end{align}
Using~\eqref{sposuse}, it follows that
\[ \|\psi_n\|_\rho^2 \rightarrow 0 \:. \]
Thus $\psi_n \rightarrow 0$ converges pointwise almost everywhere in $M$.
It follows that $\psi_n \rho \rightarrow 0$ in $\mathfrak{B}(\F)$, in contradiction to~\eqref{measurelim}.
This shows that assumption~\eqref{limNull} is false, concluding the proof.
\QED

\section{Example: A Lattice System in~$\R^{1,1} \times S^1$} \label{seclattice}
We now illustrate the previous constructions in a detailed example 
on two-dimensio\-nal Minkowski space~$\R^{1,1}$ which has some similarity
to a nonlinear sigma model with values in~$S^1$.
Furthermore, $\L$ is chosen such that the minimizer is discrete,
making the system suitable for a numerical analysis.

\subsection{The Lagrangian}
Let~$(\R^{1,1}, \la .,. \ra)$ be two-dimensional Minkowski space. Thus, denoting the
space-time points by~$\underline{x} = (x^0, x^1)$ and~$\underline{y}$, the inner product
takes the form
\[ \la \underline{x}, \underline{y} \ra = x^0 y^0 - x^1 y^1 \:. \]
Moreover, let~$\F$ be the set
\[ \F = \R^{1,1} \times S^1\:. \]
We denote points in~$x \in \F$ by $x=(\underline{x}, x^\varphi)$
with~$\underline{x} \in \R^{1,1}$ and~$x^\varphi \in [-\pi, \pi)$.
Next, we let~$A$ be the square
\[ A = (-1,1)^2 \subset \R^{1,1} \:. \]
Moreover, given~$\varepsilon \in (0,\frac{1}{4})$,
we let~$I$ be the the following subset of the interior of the light cones,
\[ I = \big\{ \underline{x} \in \R^{1,1} \, \big| \, \la \underline{x}, \underline{x} \ra > 0
\textrm{ and } |x^0| < 1 +\varepsilon \big\} \:. \]
Furthermore, we let~$f : \R^{1,1} \rightarrow \R$ be the function
\[ f(\underline{x}) = \chi_{B_\varepsilon(0,1)}(\underline{x}) + \chi_{B_\varepsilon(0,-1)}(\underline{x})
- \chi_{B_\varepsilon(1,0)}(\underline{x}) - \chi_{B_\varepsilon(-1,0)}(\underline{x}) \]
(where~$\chi$ is the characteristic function and~$B_\varepsilon$ denotes the open Euclidean ball 
of radius~$\varepsilon$ in~$\R^2 \simeq \R^{1,1}$). Finally, we let~$V : S^1 \rightarrow \R$ be the function
\[ V(\varphi) = 1 - \cos \varphi \:. \]
Given parameters~$\delta > 0$, $\lambda_I \geq 2$ and~$\lambda_A \geq 2 \lambda_I+\varepsilon$, the
Lagrangian~$\L$ is defined by
\beq
\begin{split}
\L(x,y) &= \lambda_A \ \chi_{A}( \underline x - \underline y) + \lambda_I \ \chi_{I}( \underline x - \underline y) +
V(x^\varphi -y^\varphi) \: f(\underline x - \underline y) \\
&\quad + \delta \, \chi_{B_\varepsilon(0,0)}(\underline{x}-\underline{y})\: V(x^\varphi -y^\varphi )^2 \:.
\end{split} \label{DefL}
\eeq

\begin{Lemma}
The function~$\L(x,y)$ is non-negative and satisfies the conditions~\rm{(i)} and \rm{(ii)} on page~\pageref{Cond1}.
\end{Lemma}
\Proof The only negative contributions to~$\L(x,y)$ arise in the
term~$V(x^\varphi -y^\varphi) \: f(\underline x - \underline y)$ if $\underline x - \underline y \in B_\varepsilon(1,0) \cup B_\varepsilon(-1,0)$. For $x$ and $y$ with this property, we have
\begin{align*}
V(x^\varphi -y^\varphi) \: f(\underline x - \underline y) \geq -2 \, \quad \textrm{and} \quad
\lambda_I \: \chi_{I}( \underline x - \underline y) \geq 2 \,
\end{align*}
because $V(S^1) \subset [0,2] \subset \R$,~$\lambda_I \geq 2$ and $ B_\varepsilon(1,0) \cup B_\varepsilon(-1,0) \subset I$. We conclude that~$\L(x,y) \geq 0$.

Condition~(i) is satisfied because the sets $A$ and $I$ are point-symmetric around $\underline x = 0$, $f$ is a sum of characteristic functions of sets which are mutually point-symmetric and $V(\varphi) =  V(-\varphi)$.
Condition~(ii) is satisfied because $V$ is continuous and
because characteristic functions of open sets are lower semi-continuous.
\QED

\subsection{A Local Minimizer}
We next introduce a universal measure~$\rho$ supported on the unit
lattice~$\Gamma := \Z^2 \subset \R^{1,1}$
and show that for any~$\delta > 0$, it is a local minimizer of the causal action in the sense of Definition~\ref{deflocmin}.
\begin{Lemma} The measure~$\rho$ given by
\beq
\label{DefRho}
\rho =  \sum_{\underline x \in \Gamma} \: \delta_{(\underline x,0)}
\eeq
satisfies the conditions~\rm{(iii)} and~\rm{(iv)} (on page~\pageref{Cond3}) as well as~\rm{(v)} and~\rm{(vi)} (on page~\pageref{Cond5}).
\end{Lemma}
\Proof Condition~(iii) is satisfied because for every $x \in \F$, a neighborhood $U$ with $\rho(U) < \infty$
is given for example by~$U = B_\varepsilon(\underline x)  \times  (x^\varphi -\varepsilon,x^\varphi +\varepsilon) \subset \R^{1,1} \times S^1$.
Condition~(iv) is satisfied because the function~$\L(x, .)$ is bounded and has compact support (which implies $\rho$-integrability and boundedness of $\ell$), and because $\ell$ is a finite sum of lower semi-continuous functions.

Condition~(v) is satisfied because the generalized derivative of characteristic functions exists.
Condition~(vi) holds because $\ell(x)$ is a finite sum of terms of the form $\L(x,y)$. Hence differentiation and
integration may be interchanged.
\QED

Clearly, the support of the above measure is given by
\beq \label{Mlattice}
M:= \supp \rho = \Gamma \times \{0\} \subset \F \:.
\eeq

\begin{Lemma}\label{ExStrongEL} The measure~\eqref{DefRho}
satisfies the EL equations~\eqref{ELstrong} if the parameter~$\nu$ in~\eqref{elldef} is chosen as
\beq \label{nuval}
\nu = 2 \lambda_A + 4 \lambda_I \:.
\eeq
If~$\delta>0$, the implication~\eqref{imply} holds.
\end{Lemma}
\Proof If $x \in M$, a direct computation using~\eqref{DefL} and~\eqref{elldef} shows that
\[ \ell(x) =  \lambda_A + 2 \lambda_I - \frac{\nu}{2} = 0\:. \]
Conversely, if~$x \notin M$, then either~$\underline{x} \notin \Gamma$ or~$x^\varphi \neq 0$.
In the first case, the characteristic function~$\chi_A(\underline{x} - . )$ equals one on at least two lattice points,
implying that~$\ell(x) \geq  2 \lambda_A + 2 \lambda_I - \frac{\nu}{2}=\lambda_A>0$.
In the remaining case~$\underline{x} \in \Gamma$ and~$x^\varphi \neq 0$,
the term~$\delta \: V(x^\varphi -0 )^2$ is non-negative, and it is strictly positive
if~$\delta>0$. This concludes the proof.
\QED
This lemma shows that for $\delta >0$, condition~(a) in Proposition~\ref{LocSufficient} is satisfied. The following Lemma shows that condition~(c) holds as well:
\begin{Lemma}\label{ExSpos} Choosing~$\lambda_A \geq 2 \lambda_I + \varepsilon$,
the inequality
\[ \la \psi, \L_\rho \psi \ra_\rho \geq \varepsilon\: \|\psi\|_\rho^2 \qquad
\text{holds for all~$\psi \in {\mathscr{D}}(\L_\rho)$} \:. \]
\end{Lemma}
\begin{proof}
For $\psi \in {\mathscr{D}}(\L_\rho)$ as defined by~\eqref{DomainLrho} and $e_t := (1,0) \in \R^{1,1}$, we have
\[ \big( \L_\rho \psi \big)(\underline x) = \lambda_A \psi(\underline x) + \lambda_I \big( \psi \big(\underline x + e_t \big) + \psi \big (\underline x - e_t \big) \big) \]
and hence
\beq
\la \psi, \L_{\rho} \: \psi \ra_\rho = \lambda_A \| \psi \|_\rho^2  + \lambda_I \sum_{\underline x \in \Gamma}
 \overline{\psi \big(\underline x \big)}
\Big( \psi \big(\underline x + e_t \big) + \psi \big (\underline x - e_t \big) \Big) \:. \label{LrhoEstimate}
\eeq
Applying Young's inequality
\[ \Big| \overline{\psi \big(\underline x \big)} \:\psi \big(\underline x + e_t \big) \Big| \leq \frac{1}{2} \Big(
\big| \psi \big(\underline x + e_t \big) \big|^2 + \big| \psi \big(\underline x \big) \big|^2 \Big) \:, \]
the second term of~\eqref{LrhoEstimate} can be estimated by
\[ \lambda_I \:\sum_{\underline x \in \Gamma}
 \overline{\psi \big(\underline x \big)}
\Big( \psi \big(\underline x + e_t \big) + \psi \big (\underline x - e_t \big) \Big) \geq -2 \lambda_I
\| \psi \|_\rho ^2 \:. \]
Hence~$\la \psi, \L_{\rho} \: \psi \ra_\rho \geq (\lambda_A-2 \lambda_I) \| \psi \|_\rho ^2$,
giving the result.
\end{proof}
\begin{Corollary}\label{ExLocMin}
The measure $\rho$ is a local minimizer of the causal action.
\end{Corollary}
\begin{proof} Lemma~\ref{ExStrongEL} shows that condition~(a) in Proposition~\ref{LocSufficient} holds.
Condition~(b) follows because~$\L(x,y)$ as defined by~\eqref{DefL} is bounded on $\F \times \F$.
Lemma~\ref{ExSpos} yields condition~(c).
\end{proof}

\subsection{The Jet Spaces}
We next determine the jet spaces. Clearly, in our setting of a discrete lattice~\eqref{Mlattice}, every function on~$M$
can be extended smoothly to a neighborhood of~$M$. Thus the jet space~\eqref{JM} can be written as
\[ \J|_M = \big\{ \u = (a,u) \text{ with } a : M \rightarrow \R \text{ and } u : M \rightarrow T\F \big\} \:. \]
When extending these jets to~$\F$, for convenience we always
choose an extension~$\u : \F \rightarrow T\F$ which is locally constant on~$M$.
We denote the vector component by $u=(u^0, u^1, u^\varphi)$.
In order to determine the differentiable jets, we recall from the the proof of Lemma~\ref{ExStrongEL} that
\[ \left\{ \begin{array}{ll}
\ell(\underline x, x^\varphi) =  \delta \, V(x^\varphi)^2 & \textrm{if }\underline{x} \in \Gamma \\[0.3em]
\ell(\underline x, x^\varphi) \geq \lambda_A + \delta \, V(x^\varphi)^2 &  \textrm{if } \underline{x} \notin \Gamma \:.
\end{array} \right. \]
Hence the differentiable jets~\eqref{JDiffLip} are given by
\[ \Jdiff = \big\{ \u = (a,u) \text{ with } a : M \rightarrow \R \text{ and } u=(0,0,u^\varphi) : M \rightarrow T\F \big\} \:. \]
We choose $\Jtest = \Jdiff$. 

\begin{Prp}\label{ExLin} The linearized solutions $\Jlin$ of Definition~\ref{deflin} consist of all
jets $\v = (b,v) \in \J$ with the following properties:
\begin{itemize}
\item[{\rm{(A)}}] The scalar component~$b : M \rightarrow \R$ satisfies the equation
\beq \label{ExScalComp}
\lambda_A \: b(\underline x,0) + \lambda_I\, \big( b(\underline x + e_t,0) + b(\underline x - e_t,0)
\big) = 0 \:.
\eeq
\item[{\rm{(B)}}] The vector component~$v : M \rightarrow T\F$ consists of a constant
vector~$\underline{v} \in \R^{1,1}$ and a function~$v^\varphi : M \rightarrow \R$, i.e.
\[ v(x) = \big(\underline{v}, v^\varphi(x) \big) \:, \]
where the function~$v^\varphi$ satisfies the discrete wave equation on $\Gamma$,
\beq \label{ExDiscrWaveEq}
\sum_{\underline y \in \Gamma}  f(\underline x - \underline y) \, v^\varphi\!(\underline y, 0) =  0 \, .
\eeq
\end{itemize}
\end{Prp}
\Proof For ease in notation, we identify~$M = \Gamma \times \{0\}$ with the lattice~$\Gamma$.
Our first aim is to show that that a jet~$\v \in \J$ satisfies condition~(l1) in Definition~\ref{deflin}
if and only if it is of the form
\beq \label{vform}
\v(\underline{x}) = \big(b(\underline{x}), \underline{v}, v^\varphi(\underline{x}) \big)
\eeq
with a constant vector~$\underline{v} \in \R^{1,1}$ and mappings~$b, v^\varphi : \Gamma \rightarrow \R$.

Since~(l1) does not pose a condition on the scalar component, it suffices to
consider the vector component~$v=(\underline{v}, v^\varphi)$.
Moreover, using that the Lagrangian is smooth in the variables~$x^\varphi$ and~$y^\varphi$,
it suffices to consider the component~$\underline{v}$.
If~$\underline{v}$ is a constant vector in Minkowski space, its flow does not change the
difference vector~$\underline{x}-\underline{y}$ in the Lagrangian~\eqref{DefL}.
Therefore, the combination of directional derivatives in~\eqref{derex1} exists and vanishes,
implying that the jets of the form~\eqref{vform} satisfy the condition~(l1).

The following argument shows that for every jet satisfying~(l1) the component~$\underline{v}$
is indeed constant: Assume conversely that~$\underline{v} : \Gamma \rightarrow \R^{1,1}$ is not
constant. Then there are neighboring points~$\underline{x}, \underline{y} \in \Gamma$
with~$\underline{v}(\underline{x}) \neq \underline{v}(\underline{y})$.
In order for the combination of directional derivatives in~\eqref{derex1} to exist,
the function~$\L(F_\tau(x), F_\tau(y))$ must be differentiable in~$\tau$.
This implies that the characteristic functions in~\eqref{DefL} must be continuous at~$\tau=0$.
We first evaluate this condition if~$\underline{x}$ and~$\underline{y}$ are diagonal neighbors
(i.e.\ $\underline{y} = \underline{x} + (\pm1, \pm1)$). In this case, for the characteristic
function~$\chi_I(F_\tau(x)-F_\tau(y))$ to be continuous, the
vector~$\underline{v}(\underline{x})-\underline{v}(\underline{y})$ must be
collinear to~$\underline{x}-\underline{y}$. But then the continuity of the
characteristic function~$\chi_A((F_\tau(x)-F_\tau(y))$ implies
that~$\underline{v}(\underline{x})=\underline{v}(\underline{y})$. This a contradiction.
We conclude that~$\underline{v}$ is constant on the even and odd sublattices of~$\Gamma$.
In the remaining case that~$\underline{v}$ describes a constant translation of the even sublattice relative to the
odd sublattice, we can choose neighboring lattice points~$\underline{x}, \underline{y} \in \Gamma$
(one on the odd and one on the even sublattice) such that the vector~$\underline{v}(\underline{x})
-\underline{v}(\underline{y})$ is non-zero and is not tangential to the discontinuity of the characteristic
function~$\chi_A$. This implies that the function~$\chi_A(F_\tau(x)-F_\tau(y))$ is not continuous at~$\tau=0$,
which is again a contradiction.

We conclude that condition (l1) in Definition~\ref{deflin} is satisfied precisely by all jets of the form~\eqref{vform}.
For such jets, the first part of condition (l2) is satisfied because $\L(x,.)$ has bounded support and is smooth in $x^\varphi$. It remains
to evaluate the linearized field equations~\eqref{eqlinlip}. For the constant component $\underline{v} \in \R^{1,1}$, the combinations of derivatives in~\eqref{derex1} vanishes.
Therefore we can set $\underline{v} = 0$ in the remainder of this proof.
Since all our jets are locally constant,  for $\u = (a,u) \in \Jtest$ and $x, y \in M$ we have
\begin{align*}
\nabla_{\u(x)} \nabla_{1, \v} \L(x,y) &= \Big( a(\underline x) + u^\varphi\!( \underline x) \, 
\frac{\partial}{\partial x^\varphi} \Big) \Big( b(\underline x) + 
v^\varphi\!(\underline x) \, \frac{\partial}{\partial x^\varphi}  \Big) \L(x,y)  \\
% &= \Big( a(\underline x) \,b(\underline x) +  a(\underline x) \, v^\varphi\!(\underline x) \, \frac{\partial}{\partial x^\varphi} 
% + b(\underline x) \, u^\varphi\!(\underline x) \, \frac{\partial}{\partial x^\varphi} 
%+ u^\varphi\!( \underline x) \, v^\varphi\!(\underline x)  \,\frac{\partial^2}{\partial {x^\varphi}^2}  \Big) \L(x,y) \\
&\stackrel{(\ast)}{=} \Big( a(\underline x) \,b(\underline x) 
+ u^\varphi\!( \underline x) \, v^\varphi\!(\underline x) \,\frac{\partial^2}{\partial {x^\varphi}^2}   \Big) \L(x,y) \\[0.2em]
&= a(\underline x) \,b(\underline x) \, \L(x,y) 
- u^\varphi\!( \underline x) \,v^\varphi\!(\underline x) \, f(\underline x - \underline y) \,,
\end{align*}
where in $(\ast)$ we used that the term involving the first derivative vanishes since $V(\varphi)$ is minimal
at $\varphi = 0$.  Similarly,
\begin{align} \label{ExPart}
&\nabla_{1, \u} \nabla_{2, \v} \, \L(x,y) =  a(\underline x) \,b(\underline y) \,\L(x,y) + u^\varphi\!( \underline x) \,v^\varphi\!(\underline y)  \, f(\underline x - \underline y) \:.
\end{align}
Hence, for any~$x \in M$, the linearized field equation~\eqref{eqlinlip} can be written as
\beq \label{ExWeakEL}
\begin{split}
\big( \lambda_ A + 2 \lambda_I \big) \,a(\underline{x})\, b(\underline{x}) &+
\lambda_A \,a(\underline{x}) \,b(\underline{x}) + \lambda_I \big( a(\underline{x}) \,b(\underline{x}+e_t)
+ a(\underline{x}) \,b(\underline{x}-e_t) \big) \\
&+  u^\varphi\!( \underline x) \sum_{\underline y \in \Gamma} v^\varphi\!(\underline y)  \, f(\underline x - \underline y)     - a(\underline x) \,b(\underline x) \:\frac{\nu}{2} \:,
\end{split}
\eeq
where we used the fact that~$\sum_{\underline y \in \Gamma} f(\underline x - \underline y) =0$. Evaluating~\eqref{ExWeakEL} for $\u = (a,0) \in \Jtest$ and using~\eqref{nuval} gives~\eqref{ExScalComp}.
Similarly, evaluating~\eqref{ExWeakEL} for $\u = (0,u) \in \Jtest$ yields~\eqref{ExDiscrWaveEq}. 
\QED

\subsection{The Symplectic Form}
We now construct the symplectic form. In preparation, we 
need to verify all technical assumptions.
\begin{Lemma} The conditions \emph{(s1)} to \emph{(s3)} on page~\pageref{CondS1} are satisfied.
\end{Lemma}

\begin{proof}
Condition (s1) follows because jets in $\Jlin \cap \Jtest$ only have a $v^\varphi$-component and $\L(x,y)$
is smooth in $x^\varphi$ and $y^\varphi$.
Condition (s2) follows from definition~\eqref{DefRho} and the fact that the support of~$\L(x,.)$
is finite, so that the $\rho$-integration reduces to a finite sum.
Condition (s3)  is satisfied because for $\u, \v \in \Jtest \cap \Jlin$, the commutator vanishes
due to our choice of locally constant jets.
\end{proof}

Next, in order to find an explicit expression for the symplectic form~\eqref{OSIlip}, we choose a constant time slice
\begin{align*}
N_t &= \big\{ \underline x \in \Gamma \,\big|\,  x^0 = t\big\} \quad \text{with} \quad t \in \Z\:.
\end{align*}
We let~$\Omega_{N_t}$ be the past of~$N_t$, i.e.
\[  \Omega_{N_t} = \big\{ \underline x \in \Gamma \,\big|\, x^0 \leq t \big\} \:. \]
\begin{Prp}\label{ExSympPrp}
If we choose~$\Jtest$ as the jets with spacelike compact support,
\[ \Jtest = \big\{ \u \in \Jdiff \, \big| \, \text{$\supp \u|_{N_t} $ is a finite set for all~$t \in \Z$} \big\} \:, \]
the symplectic form~\eqref{OSIlip} is given by
\begin{align}\label{ExSympl}
\begin{split}
\sigma_{\Omega_{N_t}}(\u, \v) &= \lambda_I  \sum_{\underline x \in N_t} \, \big( a(\underline x) \:b(
\underline{x} + e_t) - a(\underline{x} + e_t) \:b(\underline{x}) \big) \\
& \quad \, \, + \sum_{\underline x \in N_t} \big( \,u^\varphi\!( \underline x+e_t) \,v^\varphi\!(\underline x)   - u^\varphi\!( \underline x) \,v^\varphi\!(\underline{x}+e_t) \big) \:,
\end{split}
\end{align}
where again $e_t := (1,0) \in \R^{1,1}$.
\end{Prp} 
\begin{Remark}\em Note that the second sum~\eqref{ExSympl} is the usual symplectic form associated to
a discrete version of the wave equation on $\R^{1,1}$. Namely, after adding the terms $v^\varphi\!( \underline x) \,u^\varphi\!(\underline x )-v^\varphi\!( \underline x) \,u^\varphi\!(\underline x ) = 0$ 
to the each summand (and similarly for the scalar component)
and taking a suitable limit $e_t \rightarrow 0$, we obtain
\begin{align*}
\sigma_{\Omega_{N_t}}(\u, \v) = \sum_{\underline x \in N_t} \big( \dot u^\varphi\!(\underline x ) \,v^\varphi\!( \underline x)  - u^\varphi\!( \underline x) \,
\dot v^\varphi\!(\underline x ) \big) + \sum_{\underline x \in N_t}  \lambda_I \big( a(\underline x)\, \dot b(\underline x)
- \dot a(\underline x) \,b(\underline x)     \big) \, ,
\end{align*}
where the dot denotes the discrete $t$-derivative.
\QEDrem
\end{Remark}
 
\begin{proof}[Proof of Proposition~\ref{ExSympPrp}]
The proof of Proposition~\ref{ExLin} still goes through if~$\Jtest$ is restricted to jets with 
spatially compact support. Given $\u, \v \in \Jtest \cap \Jlin$, by applying~\eqref{ExPart} and Proposition~\ref{ExLin}, we can compute~\eqref{OSIIntegrand} to obtain
(for any choice of~$s,s' \in \{\pm\}$)
\begin{align*}
&\sigma_{\u, \v}(x,y) = \nabla_{1,\u} \nabla_{2,\v} \L(x,y) - \nabla_{1,\v} \nabla_{2,\u} \L(x,y) \\
& = a(\underline x) \,b(\underline y) \,\L(x,y) - b(\underline x) \,a(\underline y) \,\L(x,y)
+ u^\varphi\!( \underline x) \,v^\varphi\!(\underline y)  \, f(\underline x - \underline y) - v^\varphi\!( \underline x)
\,u^\varphi\!(\underline y)  \, f(\underline x - \underline y)
\end{align*}
(where we again identified~$M= \Gamma \times \{0\}$ with~$\Gamma$).
Using that
\begin{align*}
 \sigma_{\Omega_{N_t}}(\u, \v) &= \sum_{\underline x \in \Omega_{N_t}} \, \sum_{\underline y \in M \setminus \Omega_{N_t}}  
 \sigma_{\u, \v}\big((\underline x,0) ,(\underline y,0) \big) \:,
\end{align*}
we obtain~\eqref{ExSympl}.
 \end{proof}

\appendix

\section{The Fr{\'e}chet Manifold Structure of~$\calB$} \label{appclosed}
As in Section~\ref{SecSmooth} we assume that~$\F$ is a smooth manifold of dimension~$m \geq 1$.
%and that~$M:= \supp \rho$ is a smooth submanifold of~$\F$.
We want to endow the set~$\calB$ with the structure of a Fr{\'e}chet manifold.
The first step is to specify the topology on the set~$\calB$ in~\eqref{Bdef}.
We choose the {\em{compact-open topology}} defined as follows.
First, parametrizing the measures according to~\eqref{rhoFf}
by a pair~$(f, F) \in C^\infty(\F, \R) \times C^\infty(\F, \F)$, we can identify~$\calB$ with a subset of
the space of such pairs,
\[ \calB \subset C^\infty(\F, \R) \times C^\infty(\F, \F) \:. \]
Our task is to endow the sets~$C^\infty(\F, \R)$ and~$C^\infty(\F, \F)$
with the structure of a Fr{\'e}chet manifold. Once this has been accomplished,
the Fr{\'e}chet structure on~$\calB$ can be introduced simply by
assuming that~$\calB$ is a Fr{\'e}chet submanifold of the product manifold~$C^\infty(\F, \R) \times C^\infty(\F, \F)$.

Being a vector space, the space~$C^\infty(\F, \R)$ can be endowed even with the structure
of a Fr{\'e}chet space. To this end, on~$\F$ we choose an at most countable
atlas~$(x_\lambda, U_\lambda)_{\lambda \in \Lambda}$ (with an index set~$\Lambda \subset \N$)
whose charts $x_\lambda : U_\lambda \rightarrow \R^m$
are defined on relative compact subsets~$U_\lambda \subset \F$.
We then consider the Fr{\'e}chet topology induced by the~$C^k$-norms in these charts, i.e.\
\[ \| f \|_{k,\lambda} := \big\| f \circ x_\lambda^{-1} \big\|_{C^k(x_\lambda(U_\lambda))} \:. \]
The resulting topology is metrizable. It is induced for example by the distance function
\[ d(f,g) := \sum_{k=0}^\infty \;\sum_{\lambda \in \Lambda} \:2^{-k-\lambda}\, \arctan  \| f -g \|_{k,\lambda} \:. \]

In order to endow~$C^\infty(\F, \F)$ with the structure of a Fr{\'e}chet manifold,
we work locally in a neighborhood of a point~$F \in C^\infty(\F, \F)$.
First, we refine the previous atlas~$(x_\lambda, U_\lambda)$
in such a way that the domains~$U_\lambda$ are all
convex geodesic neighborhoods with respect to a chosen Riemannian metric~$g$ on~$\F$.
Moreover, we further refine this atlas to a new atlas~$(y_\gamma, V_\gamma)_{\gamma \in \Gamma}$
in such a way that~$F$ maps the domains of the new charts to domains of the old charts, meaning
that for every~$\gamma \in \Gamma$ there is~$\lambda(\gamma) \in \Lambda$ such that
\[ F(V_\gamma) \subset U_{\lambda(\gamma)} \]
(for example, one can choose the domains of the new charts as open subsets of the
sets~$U_\lambda \cap F^{-1}(U_\nu)$ for~$\lambda, \nu \in \Lambda$ and introduce the
charts as the restrictions of~$x_\lambda$ to the new domains).
We restrict attention to mappings~$G$ which are so close to~$F$ that they
map to the same charts, i.e.
\beq \label{Gcond}
G(V_\gamma) \subset U_{\lambda(\gamma)} \qquad \text{for all~$\gamma \in \Gamma$}\:.
\eeq
For such mappings, we can define the $C^k$-norms by
\[ \|G - F\|_{k, \gamma} = \big\| x_{\lambda(\gamma)} \circ G \circ y_\gamma^{-1} 
- x_{\lambda(\gamma)} \circ F \circ y_\gamma^{-1} \big\|_{C^k(y_\gamma(V_\gamma))} \:. \]
The resulting Fr{\'e}chet topology is again metrizable, as becomes obvious for example by setting
\begin{align*}
d(F,G) &= \left\{ \begin{array}{cl} 4 & \text{if~\eqref{Gcond} is violated} \\
\displaystyle 
\sum_{k=0}^\infty \;\sum_{\gamma \in \Gamma} \:2^{-k-\lambda}\, \arctan \| F-G \|_{k,\gamma}
& \text{if~\eqref{Gcond} holds}\:.
\end{array} \right.
\end{align*}
It remains to construct a local chart around~$F$. To this end, it suffices to consider mappings~$G$
which satisfy~\eqref{Gcond}. Then, since the domains of the charts~$U_\lambda$
are all geodesically convex, for any~$x \in \F$ there is a unique vector~$v(x) \in T_{F(x)} \F$
with the property that~$G(x) = \exp_{F(x)} v(x)$. In this way, the mapping~$G$ can be
described uniquely by a vector field~$v \in C^\infty(F(\F), T\F)$ on~$\F$ along~$F(\F)$.
The mapping~$G \rightarrow v$ is the desired chart, taking values in the linear space~$C^\infty(F(\F), T\F)$.

For clarity, we finally explain what the tangent vectors of~$\calB$ are, how these tangent vectors
act on functions, and how these derivatives are related to the derivative~$\nabla_{\u}$
as defined in~\eqref{Djet}. These elementary facts are also needed for the
computation of the exterior derivative~$d\gamma$ in the proof of Lemma~\ref{lemmaclosed}.
Given~$\v =(b,v) \in T_\rho {\mathcal{B}}$, we let~$\tilde{\rho}_\tau$ be a smooth curve in~$\calB$
with~$\tilde{\rho}_\tau |_{\tau = 0}=\rho$ and~$\dot{\rho}_\tau |_{\tau =0} = \v$.
We again write the measures~$\tilde{\rho}_\tau$ in the form~\eqref{paramtilrho}
(see Lemma~\ref{curve}), so that~\eqref{paramtilrho2} holds.
Then the directional derivative of a smooth function~$\phi$ on~$\calB$ is defined as usual by
\[ \v \phi = \frac{d}{d\tau} \phi\big(\tilde{\rho}_\tau \big) \big|_{\tau=0} \:. \]
In particular, the derivative of~$\gamma(\u)$ as defined in~\eqref{gammaUdef} is given by
\begin{align*}
\v \gamma(\u)\big|_{\tilde{\rho}} &= \frac{d}{d\tau}  \int_{\Omega_{N_t}} d\rho \int_{M \setminus \Omega_{N_t}}
d\rho\:
f_\tau(x)\: \nabla_{2,\u} \L\big(F_\tau(x), F_\tau(y)\big)\: f_\tau(y) \Big|_{\tau=0} \\
&= \int_{\Omega_{N_t}} d\rho \int_{M \setminus \Omega_{N_t}} d\rho\:  
\big( \nabla_{\v(x)} + \nabla_{\v(y)} \big) \nabla_{2,\u} \L(x, y) \:.
\end{align*}
We point out that here the derivative~$\nabla_{\v(y)}$ also acts on the jet~$\u$
in the derivative~$\nabla_{2,\u}$. The commutator of such products of derivatives
can be computed with the help of the following lemma.

\begin{Lemma}\label{lemComm}
For $\u, \v \in T_\rho\calB$, we have
\beq \label{commrel}
\nabla_{[\u, \v]} =  \big[ \nabla_\u, \nabla_\v \big] \:.
\eeq
\end{Lemma}
\Proof
Again denoting~ $\u=(a,u)$ and~$\v=(b,v)$, for any smooth function~$\eta$ on~$\F$ we have
\begin{align}
&\nabla_\u \nabla_\v \eta(x) = \big( a(x) + D_u \big) \big( b(x) + D_v \big) \eta(x) \qquad \textrm{and}\\
&\big[\nabla_\u, \nabla_\v \big] \eta(x) = D_{[u,v]} \eta(x) + (D_u b)(x)\, \eta(x) - (D_v a)(x)\, \eta(x)\:. \label{comm}
\end{align}
In order to compute the commutator~$[\u, \v]$, we consider
diffeomorphisms~$\Phi_\tau, \tilde{\Phi}_s :  \calB \rightarrow \calB$
along the vector fields~$\u$ and~$\v$, i.e.
\[ \Phi_0 = \text{id},\quad \partial_\tau \Phi_\tau = \u \circ \Phi_\tau \qquad \text{and} \qquad
\tilde{\Phi}_0 = \text{id},\quad \partial_s \tilde{\Phi}_s = \v \circ \tilde{\Phi}_s\:. \]
Then
\begin{align*}
\int_\F \eta(x)\, d \big(\Phi_\tau \tilde{\Phi}_s \rho \big) (x) &=
\int_\F f_\tau(x)\, \eta\big(F_\tau(x) \big) \,d \big(\tilde{\Phi}_s \rho \big)(x) \\
&= \int_\F \tilde{f}_s(x) \: f_\tau\big(\tilde{F}_s(x) \big)\, \eta\Big(F_\tau \big(\tilde{F}_s(x) \big) \Big) \,d \rho(x) 
\end{align*}
Differentiating with respect to~$s$ and~$\tau$ at~$\tau=s=0$ gives
\begin{align*}
\int_\F &\eta(x)\, d \big(\u \v \rho \big) (x) =
\frac{d^2}{d\tau ds} \int_\F \eta(x)\, d \big(\Phi_\tau \tilde{\Phi}_s \rho \big) (x) \bigg|_{s=\tau=0} \\
&= \frac{d^2}{d\tau ds} \int_\F \tilde{f}_s(x) \: f_\tau\big(\tilde{F}_s(x) \big)\, \eta\Big(F_\tau \big(\tilde{F}_s(x) \big) \Big) \,d \rho(x) \bigg|_{s=\tau=0} \\
&= \int_\F \Big( a(x) \,b(x) + (D_v a)(x) \Big)\, \eta(x)\, d\rho(x) \\
&\quad + \int_\F \Big( b(x)\: (D_u\eta)(x) + a(x)\: (D_v\eta)(x) \Big) d\rho(x) \\
&\quad + \int_\F D_v D_u \eta(x) \,  d\rho(x) \:.
\end{align*}
Likewise, exchanging the two diffeomorphism gives the vector~$\v \u \rho$. Hence
\begin{align}\label{Commutator}
 \int_\F \eta(x)\, d \big([\u, \v] \rho \big) (x) = \int_\F \Big( D_{[v, u]} \eta
+ (D_v a)\: \eta - (D_u b)\: \eta \Big) d\rho(x) \:,
\end{align}
This shows that 
\[ [\u, \v] = \big( D_u b - D_v a \, , [u, v] \big)\,. \]
Comparing with~\eqref{comm} and the definition of $\nabla_{\u}$ in~\eqref{Djet} gives~\eqref{commrel}.

We finally note for clarity that
the minus sign in~\eqref{Commutator} arises because jets $\u, \v$ act on functions on~$\calB$,
whereas the derivatives $\nabla_\u$ and~$\nabla_\v$
act on functions on~$M$. When rewriting compositions of jets $\u \v$ as compositions of derivatives
on~$M$, the order of the composition is interchanged to~$\nabla_\v \nabla_\u$. 
\QED

\Thanks {{\em{Acknowledgments:}}
We would like to thank Niky Kamran and Olaf M\"uller
for helpful discussions on jet spaces and Fr{\'e}chet manifolds as well
as Jordan Payette and the referee for valuable comments.
We are grateful to the Center of Mathematical Sciences and Applications at
Harvard University for hospitality and support. 
J.K.\ gratefully acknowledges support by the ``Studienstiftung des deutschen Volkes.''

%\bibliographystyle{amsplain}
%\bibliography{../../aarbeit/felix}
\providecommand{\bysame}{\leavevmode\hbox to3em{\hrulefill}\thinspace}
\providecommand{\MR}{\relax\ifhmode\unskip\space\fi MR }
\providecommand{\MRhref}[2]{%
  \href{http://www.ams.org/mathscinet-getitem?mr=#1}{#2}
}
\providecommand{\href}[2]{#2}

\end{document}